\documentclass[11pt]{article}

\usepackage{fancyhdr}

\usepackage{geometry}

\usepackage{thumbpdf,lmodern}
\usepackage{bm}
\usepackage{bbm}
\usepackage{float}
\usepackage{amsthm}
\usepackage{amssymb}
\usepackage{subfig}
\usepackage[round]{natbib}
\usepackage{enumerate}
\usepackage{longtable}
\usepackage{mathtools,cancel}
\usepackage{url}
\usepackage[normalem]{ulem}

\usepackage[ruled,vlined,linesnumbered,lined,boxed,commentsnumbered]{algorithm2e}
\usepackage{tikz}
\usetikzlibrary{positioning}
\usetikzlibrary{graphs} 
\usetikzlibrary[graphs]
\usepackage{multirow}
\usepackage{lscape}
\usepackage[english]{babel}
\usepackage{amsmath}
\usepackage{adjustbox}
\usepackage{tikz}
\usetikzlibrary{graphs} 
\usetikzlibrary{angles,calc,intersections,quotes,arrows.meta}
\usetikzlibrary{positioning}

\RequirePackage[OT1]{fontenc}
\RequirePackage{amsthm,amsmath}

\usepackage{authblk}

\SetKwInput{KwInput}{Input}
\SetKwInput{KwOutput}{Output}

\usepackage[linesnumbered,ruled,vlined]{algorithm2e}

\usepackage{framed}
\usepackage{xcolor}

\makeatletter
\newtheorem*{rep@theorem}{\rep@title}
\newcommand{\newreptheorem}[2]{%
	\newenvironment{rep#1}[1]{%
		\def\rep@title{#2 \ref{##1}}%
		\begin{rep@theorem}}%
		{\end{rep@theorem}}}
\makeatother

\newcommand{\BlackBox}{\rule{1.5ex}{1.5ex}}  
\ifdefined\proof
\renewenvironment{proof}{\par\noindent{\bf Proof\ }}{\hfill\BlackBox\\[2mm]}
\else
\newenvironment{proof}{\par\noindent{\bf Proof\ }}{\hfill\BlackBox\\[2mm]}
\fi
 
\newtheorem{theorem}{Theorem}
\newreptheorem{theorem}{Theorem}
\newtheorem{lemma}[theorem]{Lemma} 
\newreptheorem{lemma}{Theorem}
\newtheorem{proposition}[theorem]{Proposition} 
\newreptheorem{proposition}{Proposition}

\newtheorem{corollary}[theorem]{Corollary}
\newreptheorem{corollary}{Corollary}
\newtheorem{definition}[theorem]{Definition}

\allowdisplaybreaks 

\newcommand{\bit}{\begin{itemize}}
	\newcommand{\eit}{\end{itemize}}

\newcommand{\be}{\begin{eqnarray*}}
	\newcommand{\ee}{\end{eqnarray*}}
\newcommand{\ben}{\begin{eqnarray}}
	\newcommand{\een}{\end{eqnarray}}

\newcommand{\g}{\,\vert\,}
\newcommand{\A}{\mathcal{A}}
\newcommand{\B}{\mathcal{B}}
\newcommand{\D}{\mathcal{D}}

\newcommand{\G}{\mathcal{G}}
\newcommand{\K}{\mathcal{K}}
\newcommand{\M}{\mathcal{M}}
\newcommand{\T}{\mathcal{T}}
\newcommand{\N}{\mathcal{N}}
\newcommand{\cO}{\mathcal{O}}
\newcommand{\W}{\mathcal{W}}
\newcommand{\I}{\mathcal{I}}

\newcommand{\pa}{\mathrm{pa}}
\newcommand{\fa}{\mathrm{fa}}

\newcommand{\bzero}{\bm{0}}

\newcommand{\bA}{\bm{A}}
\newcommand{\bB}{\bm{B}}

\newcommand{\bD}{\bm{D}}

\newcommand{\bG}{\bm{G}}
\newcommand{\bI}{\bm{I}}

\newcommand{\bP}{\bm{P}}

\newcommand{\bS}{\bm{S}}
\newcommand{\bT}{\bm{T}}
\newcommand{\bU}{\bm{U}}
\newcommand{\bX}{\bm{X}}

\newcommand{\bh}{\bm{h}}

\newcommand{\bx}{\bm{x}}

\newcommand{\bSigma}{\bm{\Sigma}}

\newcommand{\bOmega}{\bm{\Omega}}

\newcommand{\bepsilon}{\bm{\varepsilon}}
\newcommand{\bphi}{\bm{\phi}}

\newcommand{\black}{\color{black}}

\newcommand{\white}{\color{white}}

\begin{document}

\title{Bayesian Causal Discovery from Unknown General Interventions}
\author[1]{Alessandro Mascaro \thanks{alessandromascaro@outlook.it}}
\author[2]{Federico Castelletti \thanks{federico.castelletti@unicatt.it}}
\affil[1]{Department of Statistical Sciences, Universit\`{a} Cattolica del Sacro Cuore, Milan}
\affil[2]{Department of Economics, Management and Statistics, Universit\`{a} degli Studi di Milano-Bicocca, Milan}

\date{}

\maketitle

\begin{abstract}
We consider the problem of learning causal Directed Acyclic Graphs (DAGs) using combinations of observational and interventional experimental data. Current methods tailored to this setting assume that interventions either destroy parent-child relations of the intervened (target) nodes or only alter such relations without modifying the parent sets, even when the intervention targets are unknown. We relax this assumption by proposing a Bayesian method for causal discovery from \emph{general} interventions, which allow for modifications of the parent sets of the unknown targets. Even in this framework, DAGs and general interventions may be identifiable only up to some equivalence classes. We provide graphical characterizations of such \textit{interventional Markov} equivalence and devise compatible priors for Bayesian inference that guarantee score equivalence of indistinguishable structures. We then develop a Markov Chain Monte Carlo (MCMC) scheme to approximate the posterior distribution over DAGs, intervention targets and induced parent sets. Finally, we evaluate the proposed methodology on both simulated and real protein expression data. 

\vspace{0.7cm}

\end{abstract}

\section{Introduction}

Directed Acyclic Graphs (DAGs) are widely used to represent causal relationships between variables. In this setting, learning the DAG structure from data is referred to as causal discovery. If only observational data are available, a DAG is in general identifiable only up to its Markov equivalence class, which includes all DAGs that imply the same conditional independencies \citep{Verma:Pearl:1990}. However, if in addition one collects interventional (experimental) data, then it is possible to identify smaller sub-classes of DAGs, known as Interventional-Markov Equivalence Classes (I-MECs) \citep{Hauser:Buehlmann:2012}. 

Current methods for causal discovery that leverage experimental data typically assume either hard or soft interventions. In essence, a \textit{hard} intervention consists of fixing the level of certain target variables and graphically corresponds to the removal of all those edges pointing towards the intervened nodes. 
On the other hand, a \textit{soft} intervention, or mechanism change \citep{Tian:Pearl:2001}, modifies the relationship between each intervened node and its parents without completely destroying it. However, these two types of interventions do not encompass the full spectrum of manipulations that an experimenter can in practice implement or achieve.

\begin{figure}
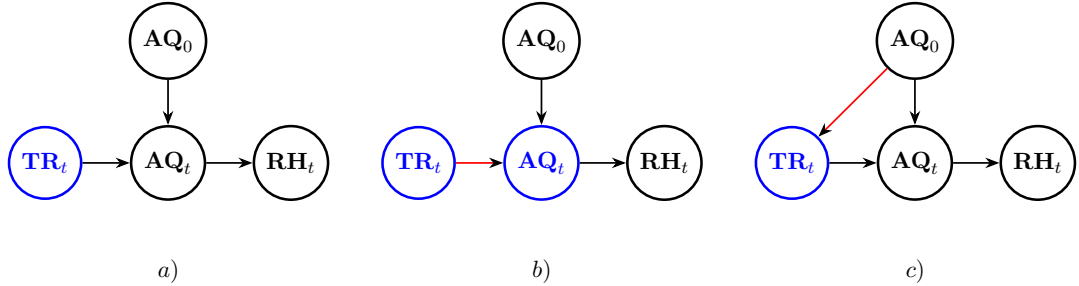

	\centering
	\begin{adjustbox}{width = 0.9\textwidth}
		\begin{tabular}{cccccccc}
			\\
			{\tikz[>={Stealth[black]}] \graph[no placement, nodes={circle, very thick, draw}, edges={black, thick}]{
					1[blue, as=$\textbf{TR}_t$,x=0,y=0] -> 2[as=$\textbf{AQ}_t$,x=2,y=0],
					4[as=$\textbf{AQ}_0$,x=2,y=2] -> 2[as=$\textbf{AQ}_t$,x=2,y=0],
					2[as=$\textbf{AQ}_t$,x=2,y=0] -> 3[as=$\textbf{RH}_t$,x=4,y=0]};}
			&& 
			{\tikz[>={Stealth[black]}] \graph[no placement, nodes={circle, very thick, draw}, edges={black, thick}]{
					1[blue, as=$\textbf{TR}_t$,x=0,y=0] ->[red] 2[blue, as=$\textbf{AQ}_t$,x=2,y=0],
					4[as=$\textbf{AQ}_0$,x=2,y=2] -> 2[as=$\textbf{AQ}_t$,x=2,y=0],
					2[as=$\textbf{AQ}_t$,x=2,y=0] -> 3[as=$\textbf{RH}_t$,x=4,y=0]};}
			&&
			{\tikz[>={Stealth[black]}] \graph[no placement, nodes={circle, very thick, draw}, edges={black, thick}]{
					1[blue, as=$\textbf{TR}_t$,x=0,y=0] -> 2[as=$\textbf{AQ}_t$,x=2,y=0],
					4[as=$\textbf{AQ}_0$,x=2,y=2] -> 2[as=$\textbf{AQ}_t$,x=2,y=0],
					4[as=$\textbf{AQ}_0$,x=2,y=2] ->[red] 1[blue, as=$\textbf{TR}_t$,x=0,y=0],
					2[as=$\textbf{AQ}_t$,x=2,y=0] -> 3[as=$\textbf{RH}_t$,x=4,y=0]};} \\ \\
			$a)$ && $b)$  && $c)$ \\
		\end{tabular}
	\end{adjustbox}
	\caption{Three DAGs resulting from different types of interventions: a) a hard intervention on $\textbf{TR}_t$; b) simultaneous hard (on $\textbf{TR}_t$) and soft (on $\textbf{AQ}_t$) interventions; c) a general intervention on $\textbf{TR}_t$. Target nodes are depicted in blue, while structural modifications induced by the interventions are colored in red.}
	\label{fig:example:on:interventions}
\end{figure}

Consider the example in Figure \ref{fig:example:on:interventions}. DAG \textit{a)} represents a causal structure involving four variables: weekly traffic level ($\textbf{TR}_t$), weekly average air quality level ($\textbf{AQ}_t$), weekly initial air quality level ($\textbf{AQ}_0$), and weekly count of individuals reporting respiratory health issues ($\textbf{RH}_t$) in a specific urban area.
In this context, a hard intervention could consist in prohibiting car access to the area, therefore setting $\textbf{TR}_t = 0$ for the subsequent weeks. A different policy might impose specific restrictions to vehicles entering the area, such as the adoption of particulate filters. This action would simultaneously reduce traffic levels and alter the relationship between traffic and air quality, thus resulting in both a hard intervention on $\textbf{TR}_t$ and a soft intervention on $\textbf{AQ}_t$; see panel \textit{b)}.
Another possible policy could regulate the number of car accesses on the basis of the initial air quality $\textbf{AQ}_0$. The resulting post-intervention graph is illustrated in panel \textit{c)} of Figure \ref{fig:example:on:interventions}, where $\textbf{AQ}_0$ is now a parent of $\textbf{TR}_t$.
This last type of intervention is commonly referred to in the literature as \textit{dynamic plan} \citep{Pearl:Robins:1995}, although sometimes still labeled as soft intervention \citep{Correa:Bareinboim:2020}. Throughout the paper, we use the term \textit{general} for those interventions that modify the parent sets of the target nodes, to emphasize their ability to represent both hard and soft interventions as special cases.

Including general interventions in a causal discovery framework becomes
essential in cases where the effect of an intervention is unknown.
For instance, in neuroimaging, and specifically in the field of effective connectivity analysis, the objective is to understand how the brain-connectivity network changes in response to external stimuli \citep{Friston:2011}. In biology, discerning key differences between gene regulatory networks may provide insights into
mechanisms of initiation and progression of specific diseases across different groups of patients \citep{Shojaie:2020}.

In this paper, we develop a Bayesian methodology for causal discovery from unknown general interventions.
We set this problem in a Bayesian model selection framework, under which priors on DAG models and associated parameters are combined with a parametric likelihood to obtain a posterior distribution on DAGs and general interventions.
Although conceptually straightforward, this task presents many challenges,
primarily the development of
\textit{compatible} parameter priors \citep{Roverato:2003:compatible} leading to closed-form DAG marginal likelihoods and guaranteeing \emph{score equivalence} for I-Markov equivalent DAGs.
Our contribution is threefold. We first provide definitions and graphical characterizations of equivalence classes of DAGs and general interventions.
We then develop a Bayesian framework for data collected under different experimental settings, which applies to 
parametric models satisfying a set of general assumptions; under the same assumptions, we develop an effective procedure for parameter prior elicitation which guarantees desirable properties in terms of marginal likelihoods, and in particular score equivalence.
Finally, we devise a Markov Chain Monte Carlo (MCMC) scheme to sample from the posterior distribution, thus allowing for posterior inference of DAG structures and general interventions. 

\subsection{Related Work}

The first historical work on causal discovery from mixtures of observational and experimental data dates back to \citet{Cooper:et:al:mixtures:1997},
who proposed a Bayesian methodology for data arising from hard interventions with known targets. Issues related to DAG identifiability in this setting were first investigated by \citet{Hauser:Buehlmann:2012}, who introduced the notion of I-Markov equivalence, provided related graphical characterizations, and developed the Greedy Interventional Equivalence Search (GIES) algorithm for structure learning. 
In the Gaussian setting, an objective Bayesian methodology working on the space of I-Markov equivalence classes was then developed by \citet{Castelletti:interventional:2019}.
In the same setting,
\citet{Wang:IGSP:2017} developed the Interventional Greedy Sparsest Permutation (IGSP) method, later extended to the case of soft interventions by
\citet{Yang:IGSP:2018}, who also
generalized the identifiability results of \citet{Hauser:Buehlmann:2012}.
An early methodology dealing with soft interventions was already proposed by \citet{Tian:Pearl:2001} who also provided graphical characterizations for Markov equivalence.

A first approach to causal discovery under \textit{uncertain} intervention targets was presented by \citet{eaton:murphy:2007}. The authors adopted a Bayesian framework for categorical data and allowed the interventions to be soft and unknown, though without addressing identifiability issues.
A more recent Bayesian methodology for Gaussian data, accounting for I-Markov equivalence and assuming hard interventions, was instead introduced by \citet{Castelletti:Peluso:2022}. In a similar setting, \citet{hagele:2023} proposed a Bayesian methodology that leverages a continuous latent representation of the posterior over DAGs and intervention targets to make use of gradient-based variational inference techniques. \citet{squires:utigsp:2020} proposed an extension of IGSP that allows for uncertainty on the targets of intervention and proved its consistency. More recently, \citet{gamella2022characterization} focused on the case of experimental Gaussian data generated from unknown noise-interventions, providing identifiability results for both DAGs and intervention targets. Similar results, in a non-parametric setting, were provided by \citet{jaber:soft:2020}, assuming soft interventions and allowing for the presence of hidden confounders. \citet{Mooij:JCI:2020} instead developed the Joint Causal Inference (JCI) framework, which encodes unknown interventions through additional indicator variables in a pooled dataset; they established under which assumptions constraint-based methods conceived for observational settings can be applied to the pooled dataset to learn the DAG and the intervention targets. 

Finally, learning the effects of unknown general interventions is equivalent to learning differences between post-intervention DAGs.
Under this perspective, our framework relates to other bodies of literature such as inference of multiple DAGs \citep{Castelletti:et:al:2020:SIM} as well as to methodologies aiming at directly estimating structural differences between causal DAGs \citep{Wang:DCI:2018}.

\subsection{Outline}

In Section \ref{sec:2:identifiability} we introduce the basic notation and background on Structural Causal Models (SCMs) and present our results relative to identifiability of DAGs and general interventions from mixtures of observational and interventional data. In Section \ref{sec:causal:discovery:general:interventions} we develop a Bayesian methodology for causal discovery in this newly defined context, leveraging the results of Section \ref{sec:2:identifiability} to provide guidance on model construction and prior elicitation. In Section \ref{sec:4:mcmc} we construct a Markov Chain Monte Carlo (MCMC) algorithm to sample from the posterior distribution of DAGs, intervention targets and induced parent sets. Finally, in Section \ref{sec:5:sim:real} we apply our methodology to the Gaussian case and empirically assess its performance on both simulated and real data. Section \ref{sec:discussion} summarizes our conclusions. All proofs of our main results are provided in the appendices to this article.
\texttt{R} code implementing our methodology is available at \url{https://github.com/alesmascaro/bcd-ugi}.

\section{Identifiability under General Interventions}
\label{sec:2:identifiability}

In this section we introduce a framework for causal discovery from unkwnon general interventions, discuss identifiability of DAGs and interventions and provide graphical characterizations of I-Markov equivalence.
Specifically, in Section \ref{sec:preliminaries} we first summarize some background material on DAGs and Structural Causal Models (SCMs) and we formalize the notion of general intervention. In Section \ref{sec:main:results} we define an I-Markov property for this new setting and present our main results on the identifiability of DAGs when interventions are known. Section \ref{sec:identifiability:unknowninterventions} extends the results to the case of unknown interventions.

\subsection{Preliminaries}
\label{sec:preliminaries}
A Directed Acyclic Graph (DAG) $\mathcal{D} = (V,E)$ with vertex set $V = [q]\coloneqq\{1,\dots,q\}$, and edge set $E \subset V\times V$ is a directed graph with no cycles, i.e.~no directed paths starting and ending at the same node. A DAG $\mathcal{D}$ can be represented by a $(q,q)$ adjacency matrix $\boldsymbol{A}$, such that $\boldsymbol{A}_{ij} = 1$ if $(i,j) \in E$
and $0$ otherwise. We let $\text{pa}_\D(j)$ be the set of \textit{parents} of node $j$, that is $\text{pa}_\D(j) = \{i \in V \g \boldsymbol{A}_{ij} = 1\}$, and $\text{fa}_\D(j) = j \cup\text{pa}_\D(j)$ be the \textit{family} of $j$ in $\D$. Moreover, an edge $i \to j$ is \textit{covered} in $\D$ if $i \cup \pa_\D(i) = \pa_\D(j)$. \black We refer to the undirected graph obtained by removing edge directions from a DAG as the \textit{skeleton} of the DAG. Any induced subgraph of the form $i \rightarrow j \leftarrow k$, with no edges between $i$ and $k$, is instead called a \textit{v-structure}. Finally, we say that $\D$ is complete if it has no missing edges.
%

Under the framework of SCMs, DAGs can be given a causal interpretation by considering each node $j$ as an observable (endogenous) variable $X_j$ and each parent-child relation as a \emph{stable} and \emph{autonomous} mechanism of the form
\begin{equation}
	\label{eq:scmdef}
	X_j = f_j(X_{\pa_\D(j)}, \varepsilon_j), \quad j \in[q],
\end{equation}
%
where $X_{\pa_{\D}(j)} = \{X_i, i\in \pa_{\D}(j)\}$, $f_j$ is a deterministic function linking $X_j$ to $X_{\pa_\D(j)}$ and to an unobserved (exogenous) random variable $\varepsilon_j$ \citep{Pearl:2000}. If $\varepsilon_1, \dots, \varepsilon_q$ are mutually independent, \black then the set of structural equations in \eqref{eq:scmdef} defines a Markovian SCM,
and the induced joint density $p(\cdot)$ on $(X_1, \dots, X_q)$ obeys the Markov property of $\D$, meaning that it factorizes as
\begin{equation}
	\label{eq:obsfactor}
	p(\bx) = \prod_{j=1}^q p(x_j \g  \bx_{\pa_\D(j)}).
\end{equation}
%
The conditional independencies implied by \eqref{eq:obsfactor} can be read-off from the DAG using the notion of \textit{d-separation} \citep{Pearl:2000}.
Let now $\M(\D)$ be the set of all positive densities $p(\bx)$ obeying the Markov property of $\D$.
Two DAGs, $\D_1$ and $\D_2$, are called \textit{Markov equivalent} if $\M(\D_1) = \M(\D_2)$.
%
DAGs can be partitioned into \textit{Markov equivalence classes}, each collecting all DAGs that are Markov equivalent. Without specific parametric assumptions, and even under common families of distributions, DAGs can be identified only up to Markov equivalence classes \citep{Pearl:1988}.
The following results provide graphical characterizations of Markov equivalence.
\begin{theorem}[\citet{Verma:Pearl:1990}]
	\label{th:markoveq}
	Two DAGs $\D_1$ and $\D_2$ are Markov equivalent if and only if they have the same skeleta and the same set of v-structures.
\end{theorem}

\begin{theorem}[\citet{chickering:1995}]
	\label{th:markoveqedge}
	Two DAGs $\D_1$ and $\D_2$ are Markov equivalent if and only if there exists a sequence $\delta$ of edge reversals modifying $\D_1$ and such that:
	\begin{enumerate}
		\item Each edge reversed is covered;
		\item After each reversal, $\D_1, \D_2$ belong to the same Markov equivalence class;
		\item After all reversals $\D_1 = \D_2$.
	\end{enumerate}
\end{theorem}
Theorem \ref{th:markoveq} provides a criterion for assessing whether two DAGs belong to the same Markov equivalence class. Theorem \ref{th:markoveqedge}, instead, is a technical result of great importance to guarantee score equivalence in score-based causal discovery methods.

The mechanisms in Equation \eqref{eq:scmdef} are stable and autonomous in the sense that it is possible to conceive an external intervention modifying one of the mechanisms (and the corresponding local distribution) without affecting the others. One can envisage different \textit{types} of external interventions \citep{Correa:Bareinboim:2020}. For any set of \textit{target} variables $T \subset [q]$ and multi-set of \textit{induced parent sets} $P=\{P_1,\dots,P_{|T|}\}$, with $P_j \subset [q]$, we consider interventions producing a mechanism change of the form
\begin{equation}
	\label{eq:mechchange}
	X_j = \tilde f_j(X_{P_j}, \varepsilon_j), \quad \forall \, j \in T.
\end{equation}
%
We refer to this type of intervention as \textit{general intervention} and, following \citet{Correa:Bareinboim:2020}, we denote the corresponding operator as $\sigma_{T,P}$.
Such intervention induces a new SCM, thus implying a new graphical object.
\begin{definition}[Post-intervention graph]
	Let $\D$ be a DAG and $(T,P)$ be a pair of intervention targets and induced parent sets defining a general intervention. The post-intervention graph of $\D$ is the graph $\D_{T,P}$ obtained by replacing for each $j \in T$ the new parents $P_j$ induced by the intervention.
\end{definition}
See also Figure \ref{fig:intervention:dag:example} for an example of DAG and implied intervention graph.
\begin{figure}[H]
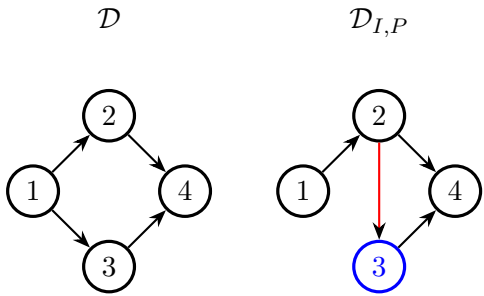

	\caption{A DAG $\D$ and the post-intervention DAG $\D_{T,P}$ for intervention target $T=\{3\}$ and induced parent set $P=\{2\}$.}
	\centering
	\begin{tabular}{ccc}
		$\D$ && $\D_{I,P}$ \\ \\
		{\tikz[>={Stealth[black]}] \graph[no placement, nodes={circle, very thick, draw}, edges={black, thick}]{
				1[x=0,y=-1] -> {3[x=1,y=-2], 2[x=1,y=0]},
				{3[x=1,y=-2], 2[x=1,y=0]} -> 4[x=2,y=-1]};}
		&& 
		{\tikz[>={Stealth[black]}] \graph[no placement, nodes={circle, very thick, draw}, edges={black, thick}]{
				1[x=0,y=-1] -> 2[x=1,y=0],
				2[x=1, y=0] ->[red] 3[blue, x=1,y=-2],
				{3[blue, x=1,y=-2], 2[x=1,y=0]} -> 4[x=2,y=-1]};}	
	\end{tabular}
	\label{fig:intervention:dag:example}
\end{figure}
Notice that a post-intervention graph need not be a DAG in general. Throughout the paper we make the following assumption, that we name \textit{validity}.
\begin{definition}[validity]
	Let $\D$ be a DAG and $(T,P)$ a pair of intervention targets and induced parent sets defining a general intervention. The general intervention is \textnormal{valid} if the post-intervention graph $\D_{T,P}$ is a DAG.
\end{definition}

As a general intervention produces a new Markovian SCM, it also induces a \textit{post-intervention} distribution through the Markov property of $\D_{T,P}$ which can be written as
\begin{align}
	\label{eq:intfactor}
	p(\bx\g\sigma_{T,P}) & = \prod_{j=1}^q \tilde{p}(x_j\g \bx_{\pa_{\D_{T,P}}(j)}) \nonumber \\
	& = \prod_{j \notin T}p(x_j\g \bx_{\pa_\D(j)}) \prod_{j \in T} \tilde{p}(x_j\g \bx_{\pa_{\D_{T,P}}(j)}),
\end{align}
where the $\tilde p(x_j\g \cdot)$'s \black denote the new local distributions induced by the intervention. For any $j \notin T$, we then have $\tilde{p}(x_j\g \bx_{\pa_{\D_{T,P}}(j)}) = p(x_j\g \bx_{\pa_\D(j)})$, so that the local densities of non-intervened nodes are invariant (stable) across pre- and post-intervention distributions. 
In the following section we show how these invariances can be leveraged to identify DAGs up to a subset of the original Markov equivalence class (named \textit{I-Markov equivalence class}) and, in the same spirit of Theorem \ref{th:markoveq} and Theorem \ref{th:markoveqedge}, we provide a graphical characterization of DAGs belonging to the same I-Markov equivalence class.

\subsection{DAG Identifiability from Known General Interventions}
\label{sec:main:results}


We consider collections of $K$ experimental settings, or environments, each defined by a general intervention with targets and induced parent sets $T^{(k)}, P^{(k)}$. Let also $\T = \{T^{(k)}\}_{k=1}^K$, $\mathcal{P} = \{P^{(k)}\}_{k=1}^K$ and $\I = (\T, \mathcal{P})$.
Each collection of experimental settings entails a family of post-intervention distributions $\big\{p(\cdot\g\sigma_{k})\big\}_{k=1}^K$, where to simplify the notation we write $\sigma_k \equiv \sigma_{T^{(k)},P^{(k)}}$ for $k \in [K]$.
We assume throughout the paper that $T^{(1)} = P^{(1)} = \O$, i.e.~$k=1$ corresponds to the observational setting where no intervention has been performed, and $p(\cdot\g\sigma_{1}) = p(\cdot)$ reduces to the pre-intervention distribution \eqref{eq:obsfactor}. Furthermore, we always assume that $\I$ is a collection of targets and induced parent sets defining a \textit{valid} general intervention.

More formally, we can define the possible tuples of joint densities corresponding to $K$ different experimental settings as follows.

\begin{definition}
	\label{def:imeclass}
	Let $\D$ be a DAG and $\I$ a collection of targets and induced parent sets. Then,
	\begin{align*}
		\M_{\I}(\D) = & \, \big\{ \{p_k(\bx)\}_{k=1}^K \ \lvert \ \forall \, k,l \in [K]: p(\bx\g \sigma_{k}) \in \M(\D_{k}) \text{ and } \\& \quad \forall \, j \notin T^{(k)} \cup T^{(l)},  p_k(x_j\g \bx_{\pa_{\D_{k}}(j)}) = p_l(x_j\g \bx_{\pa_{\D_{l}}(j)})\big\},
	\end{align*}
\end{definition}
\noindent where we let for simplicity
$p_k(\bx) = p(\bx\g\sigma_k)$ and $\D_k = \D_{T^{(k)}, P^{(k)}}$. The first condition reflects the fact that, for each experimental setting, the post-intervention distribution obeys the Markov property of the induced post-intervention DAG $\D_k$. The second condition corresponds instead to the local invariances across
post-intervention distributions of different experimental settings.
Notice that, because of the assumption $T^{(1)} = \O$, $p_1(\bx)=p(\bx)$, the observational distribution, and the condition implies that $\forall\, j \notin T^{(k)}$, $p_k(x_j\g \bx_{\pa_{\D_{k}}(j)}) = p(x_j\g \bx_{\pa_{\D}(j)})$.
By analogy with the observational case, different DAGs may still imply the same family of pre- and post-intervention distributions, leading to the notion of \textit{I-Markov equivalent} DAGs.

\begin{definition}[I-Markov equivalence]
	\label{def:imarkoveq}
	Let $\D_1$ and $\D_2$ be two DAGs and $\I$ a collection of targets and induced parent sets defining a valid general intervention for both $\D_1$ and $\D_2$. $\D_1$ and $\D_2$ are I-Markov equivalent (i.e.~they belong to the same I\textit{-Markov equivalence class}) if $\M_{\I}(\D_1) = \M_{\I}(\D_2)$.
\end{definition}

As mentioned, our aim is to develop graphical criteria to establish I-Markov equivalence between DAGs. To this end, we need: i) a graphical object that uniquely represents the DAG $\D$ and the modifications induced by the general interventions; ii) an I-Markov property to read-off the set of conditional independencies and invariances from the graphical object. 
For the first purpose, we introduce the following construction.
\begin{definition}
	\label{def:augdag}
	Let $\D$ be a DAG and $\I$ a collection of targets and induced parents sets.
	The collection of augmented intervention DAGs ($\I$-DAGs) $\{\D^{\I}_k\}_{k=1}^K$ is constructed by augmenting each post-intervention DAG $\D_k$ with an $\I$-vertex $\zeta_k$ and $\I$-edges $\{\zeta_k \to j,j \in T^{(k)}\}$.
\end{definition}
We provide an example of a collection of $\I$-DAGs in Figure \ref{fig:augmentedexample}.
The following definition extends the notion of covered edge, originally introduced by \citet[Definition 2]{chickering:1995}, to our newly defined graphical object.

\begin{definition}
	Let $\D$ be a DAG and $\I$ a collection of targets and induced parent sets implying a collection of $\I$-DAGs $\{\D^\I_k\}_{k=1}^K$. An edge $i \to j$ in $\D$ is simultaneously covered if:
	\begin{enumerate}
		\item $i \to j$ is covered in $\D$;
		\item For any $k \in [K], k\neq 1$, $i \to j$ is either covered in $\D_k^\I$, or $\{i,j\} \subseteq T^{(k)}$;.	\end{enumerate}
\end{definition}
\black 
\begin{figure}
	\caption{A collection of $\I$-DAGs for DAG $\D$ and a collection of targets and induced parent sets such that $T^{(2)} = \{3\}$, $P^{(2)} = \{1,2\}$ and $T^{(3)} = \{4\}, P^{(3)} = \{1,2,3\}$. Blue nodes represent the intervention targets, while red edges correspond to the induced parent sets.}
	\centering
	\begin{tabular}{cccccccc}
		$\D\equiv\D^{\I}_{1}$ && $\D^{\I}_{2}$ && $\D^{\I}_{3}$ \\ \\
		
		{\tikz[>={Stealth[black]}] \graph[no placement, nodes={circle, very thick, draw}, edges={black, thick}]{
				1[x=0,y=-1] -> {3[x=1,y=-2], 2[x=1,y=0]},
				{3[x=1,y=-2], 2[x=1,y=0]} -> 4[x=2,y=-1]};}
		&& 
		{\tikz[>={Stealth[black]}] \graph[no placement, nodes={circle, very thick, draw}, edges={black, thick}]{
				1[x=0,y=-1] -> 2[x=1,y=0],
				1[x=0,y=-1] ->[red] 3[blue,x=1,y=-2],
				2[x=1,y=0] ->[red] 3[blue,x=1,y=-2],
				5[as=$\zeta_2$,x=2,y=0.5] ->[dotted] 3[blue, x=1,y=-2],
				{3[blue, x=1,y=-2], 2[x=1,y=0]} -> 4[x=2,y=-1]};}
		&& 
		{\tikz[>={Stealth[black]}] \graph[no placement, nodes={circle, very thick, draw}, edges={black, thick}]{
				1[x=0,y=-1] -> {3[x=1,y=-2], 2[x=1,y=0]},
				5[as=$\zeta_3$,x=2,y=0.5] ->[dotted] 4[blue, x=2,y=-1],
				{1[x=0,y=-1], 3[x=1,y=-2], 2[x=1,y=0]} ->[red] 4[blue, x=2,y=-1]};} \\	
	\end{tabular}
	\label{fig:augmentedexample}
\end{figure}
For the second purpose instead, we introduce the following definition of I-Markov property.
\begin{definition}[I-Markov property]
	\label{def:imarkovprop}
	Let $\D$ be a DAG and $\I$ a collection of targets and induced parent sets. Let $\{p_k(\bx)\}_{k=1}^K$ be a family of strictly positive probability distributions over $(X_1,\dots,X_q)$\black. Then, $\{p_k(\bx)\}_{k=1}^K$ satisfies the I-Markov property with respect to $\{\D_k^\I\}_{k=1}^K$ if:
	\begin{enumerate}
		\item $p_k(\bx_A \g \bx_B, \bx_C) = p_k(\bx_A \g \bx_C)$ for any $k \in [K]$ and any disjoint sets $A,B,C \subset [q]$ such that $C$ d-separates $A$ and $B$ in $\D_k$;
		\item $p_k(\bx_A \g \bx_C) = p_1( \bx_A\g \bx_C)$ for any $k \in [K]$ and any disjoint sets $A,C$ such that $C$ d-separates $A$ from $\zeta_k$ in $\D^\I_k$.
	\end{enumerate}
\end{definition}

Point 1.~applies the usual Markov property to the pre- and post-intervention graphs $\D_k$, $k \in [K]$. Notice that, because general interventions may induce new parent sets, the set of implied conditional independencies may also change across experimental settings. Point 2.~instead imposes a local invariance whenever a d-separation statement involving $\I$-vertices holds in the augmented intervention DAGs. 
If a tuple of post-intervention distributions $\{p(\cdot \g \sigma_k)\}_{k=1}^K$ is $\I$-Markov w.r.t $\{\D_k^\I\}_{k=1}^K$, then any d-separation statement in $\{\D_k^\I\}_{k=1}^K$ will imply either a conditional independence relationship or an invariance in $\{p(\cdot \g \sigma_k)\}_{k=1}^K$. Throughout the paper, we also assume the converse, so that any invariance and any conditional independence relationship in the tuple of distributions implies a d-separation in $\{\D_k^\I\}_{k=1}^K$. Following \citet{squires:utigsp:2020}, we call this assumption $\I$-faithfulness.

\begin{definition}[I-Faithfulness]
	\label{ass:faithfulness}
	Let $\D$ be a DAG and a $\I$ a collection of targets and induced parent sets. Let $\{p_k(\bx)\}_{k=1}^K$ be a set of strictly positive probability distributions over $(X_1,\dots,X_q)$. Then,
	$\{p_k(\bx)\}_{k=1}^K$ is said to be I-faithful with respect to $\{\D_k^{\I}\}_{k=1}^K$ if:
	\begin{enumerate}
		\item For any $k \in [K]$ and any disjoint sets $A,B,C \subset [q]$, $p_k(\bx_A \g \bx_B, \bx_C) = p_k(\bx_A \g \bx_C)$ if and only if $C$ d-separates $A$ and $B$ in $\D_k$;
		\item For any $k \in [K]$ and any disjoint sets $A,C$, $p_k(\bx_A \g \bx_C) = p_1( \bx_A\g \bx_C)$ if and only if $C$ d-separates $A$ from $\zeta_k$ in $\D^\I_k$.
	\end{enumerate}
\end{definition}

\black

Using the I-Markov property, it is possible to characterize the newly defined I-Markov equivalence class of families of distributions through the $\I$-DAGs, as stated in the following proposition.
\begin{proposition}
	\label{prop:imarkpropok}
	Let $\D$ be a DAG and $\I$ a collection of targets and induced parent sets. Then $\{p_k(\cdot)\}_{k=1}^K \in \M_{\I}(\D)$ if and only if $\{p_k(\cdot)\}_{k=1}^K$ satisfies the I-Markov property with respect to $\{\D^\I_k\}_{k=1}^K$.
\end{proposition}
We are finally able to characterize I-Markov equivalence by means of graphical criteria.

\begin{theorem}
	\label{theorem:imarkeq:skeleta}
	Let $\D_1, \D_2$ be two DAGs and $\I$ a collection of targets and induced parent sets defining a valid general intervention for both $\D_1, \D_2$. $\D_1$ and $\D_2$ belong to the same I-Markov equivalence class if and only if $\D^\I_{1,k}$ and $\D^\I_{2,k}$ have the same skeleta and v-structures for all $k \in [K]$.
\end{theorem}

\begin{theorem}
	\label{theorem:seqcovedge}
	Let $\D_1, \D_2$ be two DAGs and $\I$ a collection of targets and induced parent sets defining a valid general intervention for both $\D_1$ and $\D_2$. $\D_1$ and $\D_2$ belong to the same I-Markov equivalence class if and only if there exists a sequence of edge reversals modifying $\D_1$ and such that:
	\begin{enumerate}
		\item Each edge reversed is simultaneously covered;
		\item After each reversal, $\{\D_{1,k}^{\I}\}_{k=1}^K$ are DAGs and $\D_1, \D_2$ belong to the same I-Markov equivalence class;
		\item After all reversals $\D_1 = \D_2$.
	\end{enumerate}
\end{theorem}

%

Theorems \ref{theorem:imarkeq:skeleta} and \ref{theorem:seqcovedge} resemble Theorems \ref{th:markoveq} and \ref{th:markoveqedge} for the observational case.
While Theorem \ref{theorem:imarkeq:skeleta} provides a direct graphical tool to assess whether two DAGs are I-Markov equivalent,
Theorem \ref{theorem:seqcovedge} is a technical result of key importance for \emph{proving} score-equivalence of DAGs.
Moreover, Theorem \ref{theorem:imarkeq:skeleta}
does not provide a characterization of I-Markov equivalence classes through a single representative graph, as \citet{Hauser:Buehlmann:2012} do for the case of hard interventions. Nevertheless, our graphical characterization is similar to the one of perfect I-Markov equivalence offered in the same paper (Theorem 10), and which is based on sequences of post-intervention DAGs.
It is thus immediate to prove the following corollary:


\begin{corollary}
	\label{cor:iessential}
	Let $\D_1$ and $\D_2$ be two DAGs and $\I$ a collection of targets and induced parent sets. $\D_1$ and $\D_2$ are I-Markov equivalent if and only if they are perfect I-Markov equivalent.
\end{corollary}

Notice however that because of our validity assumption, for a given (known) $\I$, some DAGs may be excluded from the DAG space. 
We illustrate this point with an example in Figure \ref{fig:restricted:space:example}.
\begin{figure}
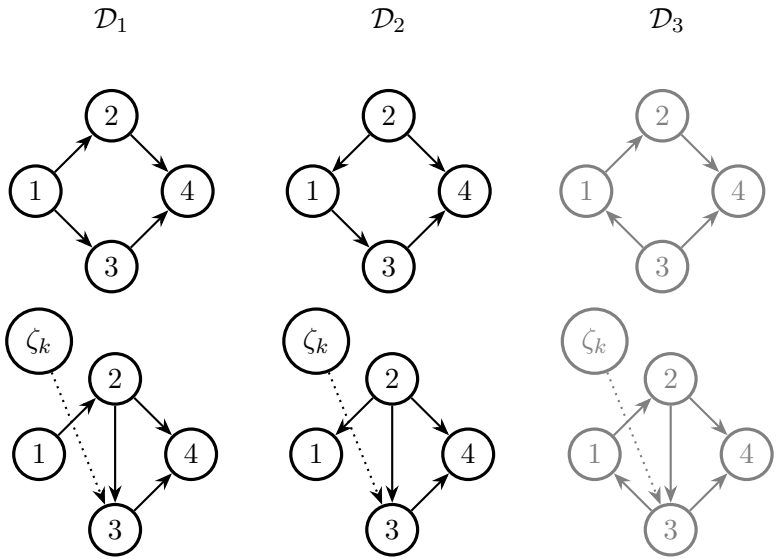

	\caption{Three Markov equivalent DAGs and their post-intervention graphs after a general intervention with $T^{(2)} = \{3\}$, $P^{(2)} = \{2\}$. The intervention is not valid for $\D_3$.}
	\centering
	\begin{tabular}{cccccc}
		$\D_1$ && $\D_2$ && $\D_3$ \\ \\
		{\tikz[>={Stealth[black]}] \graph[no placement, nodes={circle, very thick, draw}, edges={black, thick}]{
				1[x=0,y=-1] -> {3[x=1,y=-2], 2[x=1,y=0]},
				{3[x=1,y=-2], 2[x=1,y=0]} -> 4[x=2,y=-1]};}
		&&
		{\tikz[>={Stealth[black]}] \graph[no placement, nodes={circle, very thick, draw}, edges={black, thick}]{
				1[x=0,y=-1] -> 3[x=1,y=-2],
				2[x=1,y=0] -> 1[x=0,y=-1],  
				{3[x=1,y=-2], 2[x=1,y=0]} -> 4[x=2,y=-1]};}
		&& 
		{\tikz[>={Stealth[gray]}] \graph[no placement, nodes={gray, circle, very thick, draw}, edges={gray, thick}]{
				1[x=0,y=-1] -> 2[x=1,y=0],
				1[x=0, y=-1] <- 3[x=1,y=-2],
				{3[x=1,y=-2], 2[x=1,y=0]} -> 4[x=2,y=-1]};} \black \\
		{\tikz[>={Stealth[black]}] \graph[no placement, nodes={circle, very thick, draw}, edges={black, thick}]{
				5[as=$\zeta_k$,x=0,y=0.5] ->[dotted] 3[x=1,y=-2],
				1[x=0,y=-1] -> {2[x=1,y=0]},
				2[x=1,y=0] -> {3[x=1,y=-2], 4[x=2,y=-1]},
				3[x=1,y=-2] -> 4[x=2,y=-1]};}
		&&
		{\tikz[>={Stealth[black]}] \graph[no placement, nodes={circle, very thick, draw}, edges={black, thick}]{
				5[as=$\zeta_k$,x=0,y=0.5] ->[dotted] 3[x=1,y=-2],
				2[x=1,y=0] -> {1[x=0,y=-1], 3[x=1,y=-2], 4[x=2,y=-1]}, 
				3[x=1,y=-2] -> 4[x=2,y=-1]};}
		&&
		{\tikz[>={Stealth[gray]}] \graph[no placement, nodes={gray, circle, very thick, draw}, edges={gray, thick}]{
				5[as=$\zeta_k$,x=0,y=0.5] ->[dotted] 3[x=1,y=-2],
				1[x=0,y=-1] -> 2[x=1,y=0],
				2[x=1,y=0] -> {3[x=1,y=-2], 4[x=2,y=-1]},
				1[x=0, y=-1] <- 3[x=1,y=-2],
				3[x=1,y=-2] -> 4[x=2,y=-1]};} 
	\end{tabular}
	\label{fig:restricted:space:example}
\end{figure}
In such case, the general intervention defined by $T^{(2)} = {3}, P^{(2)} = {2}$ is valid for $\D_1$ and $\D_2$, but not for $\D_3$, as it would induce a cycle. Accordingly, if we consider the equivalence class defined by this intervention and assume its validity, then node $2$ cannot be a descendant of node $3$. This implies that DAGs for which $2$ is instead a descendant of $3$ must be excluded from the original DAG space.
While this implication may appear undesirable, it is worth noting 
that it only occurs when the intervention targets are \emph{known}, and the intervention includes the addition of a new parent node. In the next section we instead consider the case of \emph{unknown} interventions, thus avoiding the assumption of known targets and induced parent sets.

\subsection{DAG Identifiability from Unknown General Interventions}
\label{sec:identifiability:unknowninterventions}

In the previous section we introduced I-Markov equivalence as a limit to DAG identifiability from a collection of experimental settings characterized by \textit{known} targets and induced parent-sets $(\T, \mathcal{P})$.
In this section, we consider the problem of jointly identifying the the pair $(\D, \I)$ from a family of pre- and post-intervention distributions $\{p(\cdot \g \sigma_k)\}_{k=1}^K$. 
The same problem has been previously investigated by \cite{squires:utigsp:2020} in the context of soft interventions. The authors showed that, assuming $\I$-faithfulness, the DAG identifiability limit remains the same even when the targets of intervention are unknown and must be learnt from the data. Their results only partially apply to our general intervention setting, and accordingly further considerations are required. 
We first consider the problem of learning a general intervention from a known DAG $\D$ and a given family of distributions $\{p_k(\cdot)\}_{k=1}^K$. Any general intervention induces a sequence of augmented DAGs that, through the I-Markov property of Definition \ref{def:imarkovprop}, implies a set of conditional independencies and invariances. We thus investigate the limits in the identifiability of $(\T, \mathcal{P})$, that is whether different general interventions may imply the same set of conditional independencies and invariances.
With a slight abuse of terminology, we will refer to indistinguishable general interventions as I-Markov equivalent.

\begin{definition}
	\label{def:tpimarkoveq}
	Let $\D$ be a DAG and $\I_1, \I_2$ two collections of targets and induced parent sets. $\I_1, \I_2$ are I-Markov equivalent (or, equivalently, belong to the same \textit{I-Markov equivalence class}) if $\M_{\I_1}(\D) = \M_{\I_2}(\D)$.
\end{definition} 

Consider for instance the two general interventions depicted in Figure \ref{fig:equivalent:interventions}, where we have $T_1^{(2)} = T_2^{(2)} = \{1,3\}$, $P_1^{(2)} = \{\{3\}, \O\}$ and $P_2^{(2)} = \{\O, \{1\}\}$. In both cases, the pre- and post-intervention DAGs have the same skeleta and the same set of v-structures, thus implying the same d-separation statements. As a consequence, also the conditional independencies and invariances are the same and the two general interventions are indistinguishable given data alone.
\begin{figure}
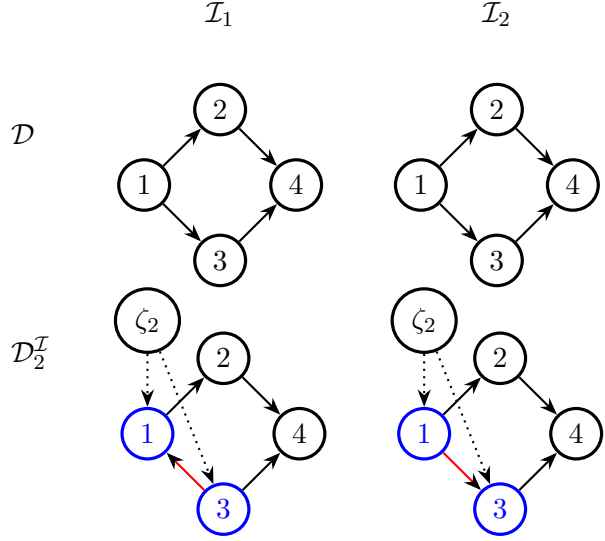

	\caption{Two non-identifiable combinations of DAGs and general interventions.}
	\centering
	\vspace{0.5cm}
	\begin{tabular}{lccccc}
		$\textbf{}$ && $\I_1$ && $\I_2$ \\ \\
		\multirow{1}{*}{$\D$}
		&&
		\raisebox{-0.7\height}{\tikz[>={Stealth[black]}] \graph[no placement, nodes={circle, very thick, draw}, edges={black, thick}]{
				1[x=0,y=-1] -> {3[x=1,y=-2], 2[x=1,y=0]},
				{3[x=1,y=-2], 2[x=1,y=0]} -> 4[x=2,y=-1]};}
		&&
		\raisebox{-0.7\height}{\tikz[>={Stealth[black]}] \graph[no placement, nodes={circle, very thick, draw}, edges={black, thick}]{
				1[x=0,y=-1] -> {3[x=1,y=-2],2[x=1,y=0]},
				{3[x=1,y=-2], 2[x=1,y=0]} -> 4[x=2,y=-1]};} \\
		\multirow{1}{*}{$\D_2^\I$}
		&& 
		\raisebox{-0.7\height}{\tikz[>={Stealth[black]}] \graph[no placement, nodes={circle, very thick, draw}, edges={black, thick}]{
				1[blue,x=0,y=-1] -> 2[x=1,y=0],
				1[x=0, y=-1] <-[red] 3[blue, x=1,y=-2],
				5[as=$\zeta_2$,x=0,y=0.5] ->[dotted] {1[blue,x=0,y=-1], 3[blue, x=1,y=-2]},
				{3[blue, x=1,y=-2], 2[x=1,y=0]} -> 4[x=2,y=-1]};} 
		&& 
		\raisebox{-0.7\height}{\tikz[>={Stealth[black]}] \graph[no placement, nodes={circle, very thick, draw}, edges={black, thick}]{
				1[blue,x=0,y=-1] -> 2[x=1,y=0],
				1[x=0, y=-1] ->[red] 3[blue, x=1,y=-2],
				5[as=$\zeta_2$,x=0,y=0.5] ->[dotted] {1[blue,x=0,y=-1], 3[blue, x=1,y=-2]},
				{3[blue, x=1,y=-2], 2[x=1,y=0]} -> 4[x=2,y=-1]};}
	\end{tabular}
	\label{fig:equivalent:interventions}
\end{figure}
We then provide the following characterizations of I-Markov equivalence of general interventions.

\begin{theorem}
	\label{theorem:imarkeq:skeleta2}
	Let $\D$ be a DAG and $\I_1, \I_2$ two collections of targets and induced parent sets. Then, $\I_1, \I_2$ belong to the same I-Markov equivalence class if and only if $\D^{\I_1}_{k}, \D^{\I_2}_{k}$ have the same skeleta and v-structures for all $k \in [K]$.
\end{theorem}

\begin{theorem}
	\label{theorem:seqcovedge2}
	Let $\D$ be a DAG and $\I_1$, $\I_2$ two collection of targets and induced parent sets. Then, $\I_1$, $\I_2$ belong to the same I-Markov equivalence class if and only if for each $\I$-DAG $\D_k^{\I_1}$ there exists a sequence of edge reversals modyfing $\D_k^{\I_1}$ and such that:
	\begin{enumerate}
		\item Each edge reversed is covered;
		\item After each reversal, $\D_{k}^{\I_1}$ is a DAG and $\I_1$, $\I_2$ belong to the same I-Markov equivalence class;
		\item After all reversals $\D_{k}^{\I_1} = \D_{k}^{\I_2}$.
	\end{enumerate}
\end{theorem}

I-Markov equivalent general interventions thus imply the same skeleta in $\{\D_k^\I\}_{k=1}^K$, and in particular, the same sets of $\I$-edges in the augmented DAGs. This implies that the intervention targets are identifiable. 

We now consider the problem of \emph{jointly} identifying $(\D, \I)$, that is the DAG and the collection of targets and induced parent sets. 
As before, we will use the term I-Markov equivalent to refer to indistinguishable pairs $(\D_1, \I_1)$ and $(\D_2, \I_2)$.




\begin{definition}
	\label{def:dtpimarkoveq}
	Let $\D_1, \D_2$ be two DAG and $\I_1, \I_2$ two collections of targets and induced parent sets defining a valid general intervention for $\D_1, \D_2$ respectively. $(\D_1, \I_1), (\D_2, \I_2)$ are I-Markov equivalent (or, equivalently, belong to the same \textit{I-Markov equivalence class}) if $\M_{\I_1}(\D_1) = \M_{\I_2}(\D_2)$.
\end{definition}

As before, we now provide graphical characterizations of I-Markov equivalence for $(\D, \I)$. 


\begin{theorem}
	\label{theorem:imarkeq:skeleta3}
	Let $\D_1, \D_2$ be two DAGs and $\I_1, \I_2$ two collections of targets and induced parent sets defining a valid general intervention for $\D_1, \D_2$ respectively. $(\D_1, \I_1), (\D_2, \I_2)$ belong to the same I-Markov equivalence class if and only if $\D^{\I_1}_{1,k}, \D^{\I_2}_{2,k}$ have the same skeleta and v-structures for all $k \in [K]$.
\end{theorem}

\begin{theorem}
	\label{theorem:seqcovedge3}
	Let $\D_1, \D_2$ be two DAGs and $\I_1$, $\I_2$ two collections of targets and induced parent sets defining a valid general intervention for both $\D_1, \D_2$. $(\D_1, \I_1)$, $(\D_2, \I_2)$ belong to the same I-Markov equivalence class if and only if there exists a sequence of edge reversals modifying the collection of $\I$-DAGs $\{\D_{1,k}^{\I_1}\}_{k=1}^K$ and such that:
	\begin{enumerate}
		\item Each edge reversed in $\D_1$ is simultaneously covered;
		\item Each edge reversed in $\D_{1,k}^{\I_1}$, for $k \neq 1$, is covered;
		\item After each reversal, $\{\D_{1,k}^{\I_1}\}_{k=1}^K$ are DAGs and $(\D_1,\I_1)$, $(\D_2, \I_2)$ belong to the same I-Markov equivalence class;
		\item After all reversals $\D_{1,k}^{\I_1} = \D_{2,k}^{\I_2}$ for each $k \in [K]$.
	\end{enumerate}
\end{theorem}

As before, by Theorem \ref{theorem:imarkeq:skeleta3}, two distinct I-Markov equivalent pairs $(\D_1, \I_1)$, $(\D_2, \I_2)$ have the same set of $\I$-edges, meaning that $\T_1 = \T_2$ and the targets are identifiable from the data. $\I_1, \I_2$ thus differ for their induced parent sets, and in particular for the reversal of covered edges connecting two target nodes. 
Note in addition that the graphical criterion of Theorem \ref{theorem:imarkeq:skeleta3} is equivalent to the one of Theorem \ref{theorem:imarkeq:skeleta}. As a consequence, any two non-identifiable pairs $(\D_1, \I_1)$, $(\D_2, \I_2)$ imply the same set of conditional independencies and invariances via the I-Markov property and in particular the same as if the general interventions were known. The DAG-identifiability limit thus remains the same as for the known intervention case. 

\black

\section{Bayesian Causal Discovery}
\label{sec:causal:discovery:general:interventions}

In this section we introduce a parametric Bayesian framework for the analysis of data collected under general unknown interventions.
In Section \ref{sec:bayesian:causal:discovery:likelihood}
we frame the related causal discovery
problem
under the Bayesian perspective, and specify a likelihood function that integrates data from distinct interventional contexts.
In Section \ref{sec:parameter:priors:elicitation} we then introduce a prior elicitation procedure for the collection of model parameters.
Finally, in Section \ref{sec:priors:model}
we assign prior distributions to DAGs, intervention targets and parent sets, whose posterior inference represents the ultimate goal of our Bayesian methodology.

\subsection{Model Formulation}
\label{sec:bayesian:causal:discovery:likelihood}

Let $\bX = \big(\bX^{(1)},\dots,\bX^{(K)}\big)^\top$ be an $(n,q)$ data matrix, such that $\bX^{(k)}$ is the $(n_k, q)$ dataset containing samples collected under the \textit{k}-th experimental setting. As in the previous sections, we assume $\bX^{(1)}$ being an observational dataset, so that $T^{(1)} = P^{(1)} = \O$ and $\D_1 = \D$.
Under the Bayesian setting, learning the pair $(\D, \I)$
can be framed as a model selection problem which requires the computation of the posterior distribution
\begin{align}
	\label{eq:postdistr}
	\begin{split}
		p\left(\D, \I \g \bX \right) 
		&  \propto p\left(\bX \g \D, \I \right)p(\D, \I). 
	\end{split}
\end{align}
We refer to $p(\D, \I)$ as the \textit{model prior} and to $p(\bX \g \D, \I)$ as the \textit{model evidence} or \textit{marginal likelihood}. Assuming a parametric family of distributions for the
observables, we can write the marginal likelihood as
\begin{equation}
	\label{eq:marglikint}
	p\big(\bX \g \D, \I\big) = \int p\big(\bX \g \Theta^{(\mathcal{K})}, \D, \I\big) \, p\big(\Theta^{(\K)} \g \D, \I\big) \,d\Theta^{(\K)},
\end{equation}
where $\Theta^{(\K)}=\{\Theta^{(1)}, \dots, \Theta^{(K)}\}$ is the multi-set of parameters associated with the pre- and post-intervention distributions implied by the pair $(\D, \I)$. Conditionally on $\Theta^{(\mathcal{K})}$, the observations in $\bX$ are independent and, within each block $\bX^{(k)}$, identically distributed, so that the likelihood function can be written as
\begin{equation}
	\label{eq:likelihood}
	p\big(\bX \g \Theta^{(\mathcal{K})}, \D, \I\big) = \prod_{k=1}^K p\big(\boldsymbol{X}^{(k)} \g \Theta^{(k)}, \D,I^{(k)}\big),
\end{equation}
where $I^{(k)} = (T^{(k)}, P^{(k)})$ and $\Theta^{(k)}$ is the set of parameters of the distribution of the $k$-th experimental setting. From Definition \ref{def:imarkovprop}, the I-Markov property implies that: i) the \textit{sampling distribution} of the $i$-th observation in the $k$-th block factorises according to the post-intervention DAG $\D_k$; ii) a set of invariances hold, such that the post-intervention local parameters indexing the non-intervened nodes are equal to the corresponding pre-intervention parameters. From these considerations, it follows that
\begin{align}
	\label{eq:likelihoodfinal}
	\begin{split}
		p\big(\bX \g \Theta^{(\mathcal{K})}, \D, \I\big) & =
		\prod_{j=1}^q \Bigg\{ p\big(\boldsymbol{X}_{\cdot j}^{\mathcal{A}(j)} \g \boldsymbol{X}_{\cdot\pa_{\D}(j)}^{\mathcal{A}(j)}, \Theta_j^{(1)}, \D \big) \\ & \qquad \quad \prod_{k: j \in T^{(k)}} p\big(\boldsymbol{X}_{\cdot j}^{(k)} \g \boldsymbol{X}_{\cdot\pa_{\D_k}(j)}^{(k)}, \Theta_j^{(k)}, \D_k \big) \Bigg\},
	\end{split}
\end{align}
where $\Theta_j^{(k)}$ is the $j$-th element of $\Theta^{(k)}$, and we denote the conditioning on $(\D, I^{(k)})$ through the modified DAG $\D_k$. Moreover, $\mathcal{A}(j) := \{k: j \notin T^{(k)}\}$ is the collection of interventional settings under which node $j$ has not been intervened upon, and 
$\bX_{.B}^{\A(j)}$ is the sub-matrix of $\bX$ with columns indexed by $B\subset [q]$ and blocks corresponding to $\A(j)\subset [K]$.
To obtain \eqref{eq:postdistr} we thus need to specify:
\begin{enumerate}
	\item A \textit{statistical model} $p\left(\bX \g \Theta^{(\mathcal{K})}, \D, \I \right)$, in the form of a distribution for the data in Equation \eqref{eq:likelihoodfinal}; \black
	\item A \textit{model prior} $p(\D, \I)$, describing our prior knowledge on DAG $\D$ and on the effects that the interventions imply on its structure; 
	\item A \textit{parameter prior} $p(\Theta^{(\mathcal{K})} \g \D, \I)$ leading, once combined with the likelihood \eqref{eq:likelihoodfinal}, to the marginal likelihood \eqref{eq:marglikint}.
\end{enumerate}
The joint specification of a statistical model and associated parameter prior deserves particular attention and is the main subject of the next section.

\subsection{Parameter Prior Elicitation}
\label{sec:parameter:priors:elicitation}

Under common distributional assumptions (e.g.~Gaussian), it is not possible to distinguish between DAGs belonging to the same I-Markov equivalence class \citep{Hauser:Buehlmann:2012}. In a Bayesian model-selection framework, this feature translates into the compatibility requirement that I-Markov equivalent DAGs are assigned equal marginal likelihoods, a property usually referred to as \textit{score equivalence}. In this section we show how the procedure proposed by \citet{Geiger:Heckerman:2002} for DAG model selection from observational data can be extended to our interventional setting.
\black
Their methodology
relies on a set of assumptions (Assumptions 1-5 in the original paper) that translate into our setting as follows:

\begin{itemize}
	\item[\textbf{A1}] \textit{(Complete model equivalence and regularity)}: Let $\mathcal{C}$ be the collection of complete DAGs on the set of nodes $V$, each implying a statistical model $p(\bx\g\Theta_{C}, C)$, for $C\in\mathcal{C}$. For any two complete DAGs $C_i, C_j \in \mathcal{C}, i \neq j$, we have that $p(\bx \g \Theta_{C_i}, C_i) = p(\bx \g \Theta_{C_j}, C_j)$. \black Moreover, there exists a one-to-one mapping $\kappa_{i,j}$ between the DAG-parameters $\Theta_{C_i}, \Theta_{C_j}$ such that $\Theta_{C_j} = \kappa_{i,j}(\Theta_{C_i})$ and the Jacobian $|\partial \Theta_{C_i}/\partial \Theta_{C_j}|$ exists and is nonzero for all values of $\Theta_{C_i}$;
	\item[\textbf{A2}] \textit{(Likelihood and prior modularity)}: For any two DAGs $\D_i, \D_j$ and any node $l \in V$ such that $\pa_{\D_i}(l) = \pa_{\D_j}(l)$, we have that, for any collection of targets and induced parent sets $\I$, \\ $$p\big(\bx_l^{(k)} \g \bx_{\pa_{\D_{i,k}}(l)}^{(k)}, \Theta_l^{(k)}, \D_{i,k}\big) = p\big(\bx_l^{(k)} \g \bx_{\pa_{\D{j,k}}(l)}^{(k)}, \Theta_l^{(k)}, \D_{j,k}\big),$$ $$p\big(\Theta_l^{(k)}\g\D_{i,k}\big) = p\big(\Theta_l^{(k)}\g\D_{j,k}\big);$$
	\item[\textbf{A3}] \textit{(Global parameter independence)}: For every DAG $\D$ and any collection of targets and induced parent sets $\I$, $$p\big(\Theta^{(\mathcal{K})}\g\D, \I\big) = \prod_{j=1}^q\left\{ p\big(\Theta_j^{(1)}\g\D\big)\prod_{k:j \in T^{(k)}} p\big(\Theta_j^{(k)}\g\D_k\big)\right\}.$$
\end{itemize}
We refer the reader to \citet{Geiger:Heckerman:2002} for a detailed discussion of these assumptions in the observational setting. Most importantly for our purposes, given Assumption \textbf{A3}, we can specify priors for the parameters indexing each term in \eqref{eq:likelihoodfinal} independently. The following procedure is therefore applied to each node $j \in V$ and experimental context $k \in [K]$:
\begin{itemize}
	\item[\textbf{i)}] Identify a complete DAG $C_{j,k}$ such that $\pa_{C_{j,k}}(j) = \pa_{\D_k}(j)$;
	\item[\textbf{ii)}] Assign a prior to $\Theta_{C_{j,k}}$, the parameter of the selected complete DAG model $C_{j,k}$;
	\item[\textbf{iii)}] Assign to $\Theta_j^{(k)}$ the same prior assigned to $\Theta_{j,C_{j,k}}$ in step \textbf{ii)}, where $\Theta_{j,C_{j,k}} \in \Theta_{C_{j,k}}$ is the parameter indexing the $j$-th node.
\end{itemize}
Accordingly, because of Assumption \textbf{A1}, the proposed procedure allows to specify a parameter prior for any pair $(\D, \I)$ from a single  parameter prior on a complete DAG model $C$. Therefore, the marginal likelihood $p\left(\bX \g \D, \I\right) $ can be computed as in the following proposition.
\begin{proposition}
	\label{prop:marglik}
	Given any complete DAG $C$ and a data matrix $\bX$ collecting observations from $K$ different experimental settings, for any valid pair $(\D, \I)$ Assumptions \textbf{A1}-\textbf{A3} imply
	\begin{align}
		\label{eq:marg:like:general}
		\begin{split}
			p\left(\bX \g \D, \I \right) & =
			\prod_{j=1}^q \left\{ \frac{p\left(\boldsymbol{X}_{\cdot \fa_{\D}(j)}^{\mathcal{A}(j)} \g C \right)}{p\left(\boldsymbol{X}_{\cdot \pa_{\D}(j)}^{\mathcal{A}(j)} \g C \right)} \prod_{k: j \in T^{(k)}} \frac{p\left(\boldsymbol{X}_{\cdot \fa_{\D_k}(j)}^{(k)} \g C \right)}{p\left(\boldsymbol{X}_{\cdot \pa_{\D_k}(j)}^{(k)} \g C \right)} \right\},
		\end{split}
	\end{align}
	where $p\big(\boldsymbol{X}_{\cdot B}^{\mathcal{A}(j)} \g C \big)$ is the marginal data distribution computed under any complete DAG $C$.
\end{proposition}

Notice that the resulting marginal likelihood provides a \emph{decomposable} score for the pair $(\D,\I)$, since it corresponds to a product of $q$ terms each involving a node $j$ and its parents $\pa_{\D_k}(j)$ in each DAG $\D_k$ only.
Importantly, it also guarantees score equivalence for I-Markov equivalent pairs $(\D,\I)$.
\vspace{0.2cm}

\begin{theorem}[Score equivalence]
	\label{thm:scoreequivalence}
	Let $\D_1, \D_2$ be two DAGs and $\I_1, \I_2$ two collections of targets and induced parent sets defining a valid general intervention for $\D_1, \D_2$ respectively. If $(\D_1,\I_1)$ and $(\D_2,\I_2)$ are I-Markov equivalent, then Assumptions A1-A3 imply 
	\begin{equation}
		p(\bX\g\D_1, \I_1) = p(\bX\g\D_2, \I_2).
	\end{equation}
\end{theorem}

\subsection{Prior on $(\D,\I)$}
\label{sec:priors:model}

Recall that $\I = (\T,\mathcal{P})$, where $\T=\{T^{(k)}\}_{k=1}^K$ and $\mathcal{P}=\{P^{(k)}\}_{k=1}^K$.
For convenience, we represent the (possibly) different parent sets induced by the $K$ experimental settings, $\mathcal{P}$, through $K$ $(q,q)$ matrices $\bP^{(1)},\dots,\bP^{(K)}$ such that for any $(l,j)$-element $\bP_{lj}^{(k)}$ we have
$\bP_{lj}^{(k)} = 1$ if $l\rightarrow j \in \D_k$ and $j \in T^{(k)}$, $0$ otherwise.
%
Conditionally on DAG $\D$ and target $T^{(k)}$, we assume independently across $k \in \{2, \dots, K\}$,
\ben
\label{eq:prior:parent:sets}
\begin{aligned}
	p\big(\bP^{(k)}\g \bphi^{(k)}, T^{(k)}, \D\big)
	\,\,&=\,\,
	\left\{
	\prod_{j=1}^q
	\prod_{j \in T^{(k)}}
	\textnormal{pBern}
	\big(\bP_{lj}^{(k)}\g \phi_j^{(k)}\big)
	\right\}
	\mathbbm{1}
	\left\{
	\D_k \textnormal{ is a DAG}
	\right\} \\
	\phi_j^{(k)}
	\,\,&\stackrel{\textnormal{iid}}{\sim}\,\,
	\textnormal{Beta}\big(a_{\phi},b_{\phi}\big), \quad j \in T^{(k)},
\end{aligned}
\een
where $\bphi^{(k)}=\big\{\phi_j^{(k)}\big\}_{j \in T^{(k)}}$.
The hierarchical prior \eqref{eq:prior:parent:sets} leads to the marginal (integrated w.r.t.~$\bphi^{(k)}$) prior on $\bP^{(k)}$
\be
p\big(\bP^{(k)}\g T^{(k)}, \D\big)
\,=\,
\left\{\prod_{j\in T^{(k)}}
\frac
{\B\left(a_{\phi} + |\bP_{.j}^{(k)}|,b_{\phi}+q-|\bP_{.j}^{(k)}|\right)}
{\B\big(a_{\phi},b_{\phi}\big)}
\right\}
\mathbbm{1}
\left\{
\D_k \textnormal{ is a DAG}\right\},
\ee
where $|\bP_{.j}^{(k)}|=\sum_{l=1}^q\bP_{lj}^{(k)}$ and $\B(\cdot)$ denotes the Beta function.

\vspace{0.3cm}

\noindent
Now consider $T^{(k)}$, the intervention target associated with the experimental setting $k$.
We represent $T^{(k)}\subseteq[q]$ through a $(q,1)$ vector $\bh_k$ whose $j$-th element $h_k(j)$ is equal to $1$ if $j \in T^{(k)}$, $0$ otherwise.
We assume, independently across $k \in \{2, \dots, K\}$,
\ben
\label{eq:prior:target}
\begin{aligned}
	p\left(\bh_k\g \eta_k\right)
	\,\,&=\,\,
	\prod_{j=1}^q
	\textnormal{pBern}
	\left(h_k(j)\g \eta_k\right)
	\\
	\eta_k
	\,\,&\sim\,\,
	\textnormal{Beta}\left(a_{\eta},b_{\eta}\right).
\end{aligned}
\een
Equation \eqref{eq:prior:target} leads to the integrated prior on $T^{(k)}$
\be
p\big(T^{(k)}\big)
\,=\, p(\bh_k)
\,=\,
\frac
{\B\big(a_{\eta} + |T^{(k)}|,b_{\eta}+q-|T^{(k)}|\big)}
{\B\big(a_{\eta},b_{\eta}\big)},
\ee
where $|T^{(k)}|=\sum_{j=1}^q h_k(j)$ is the number of intervened nodes in context $k$.

\vspace{0.3cm}

\noindent
Finally, let $\mathcal{S}_q$ be the set of all DAGs with $q$ nodes. We assign a prior to $\D\in\mathcal{S}_q$ through a collection of Bernoulli random variables indicating the absence/presence of links in the graph.
Specifically, let $\bS^{\D}$ be the adjacency matrix of the skeleton of $\D$, and $\bS_{lj}^{\D}$ its $(l,j)$-element.
We assign
\ben
\label{eq:prior:dag}
\begin{aligned}
	p\big(\bS^{\D} \g \pi\big) \,\,&=\,\, \prod_{l<j}\text{pBern}\big(\bS_{lj}^{\D}\g\pi\big) \\
	\pi \,\,&\sim\,	\, \textnormal{Beta}(a_{\D},b_{\D}),
\end{aligned}
\een
leading to
\be
p(\bS^{\D}) =
\frac
{\B\big(a_{\D} + |\bS^{\D}|,b_{\D}+q(q-1)/2-|\bS^{\D}|\big)}
{\B\big(a_{\D},b_{\D}\big)},
\ee
where $|\bS^{\D}|$ is the number of edges in $\D$ (equivalently in its skeleton) and $q(q-1)/2$ is the maximum number of edges in a DAG on $q$ nodes.
Finally, we set
$p(\D)\propto p(\bS^{\D})$ for each $\D\in\mathcal{S}_q$.

\vspace{0.5cm}

\black

\section{MCMC Scheme and Posterior Inference}
\label{sec:4:mcmc}

In this section we describe the Markov Chain Monte Carlo (MCMC) strategy that we adopt to approximate the posterior distribution
\eqref{eq:postdistr}. Specifically, Section \ref{sec:MCMC:sampling:scheme} introduces the random scan Metropolis-Hastings algorithm which is at the basis of our sampler, while Section \ref{sec:posterior:inference} illustrates how the MCMC output can be used to provide estimates of the underlying causal DAG structure and the effects of the general interventions.

\subsection{Sampling Scheme}
\label{sec:MCMC:sampling:scheme}

Our MCMC algorithm has the structure of a random-scan component-wise Metropolis-Hastings \citep[Chapter 1]{Brooks:et:al:book:2011}, in which the
parameter of interest is partitioned into $K$ components, each indexing one of the $K$ experimental settings. Specifically, the first component corresponds to the DAG $\D$, while the remaining ones to the collection of unknown targets and induced parent sets $I^{(k)} = (T^{(k)}, P^{(k)})$ for $k \in \{2, \dots, K\}$. 
Sampling from each component occurs in a random order through standard proposal and acceptance/rejection steps as in a Metropolis-Hastings sampler. A high-level illustration of the scheme is provided in Algorithm \ref{alg:mhwithingibbs}.


\begin{algorithm}{
		\SetAlgoLined
		\vspace{0.1cm}
		\KwInput{Data matrix $\bX$, number of MCMC iterations $S$, initial values for DAG, targets and induced parent sets $\D^0, \T^0, \mathcal{P}^0$}
		\KwOutput{$S$ samples from $p(\D, \T, \mathcal{P} \g \bX)$}
		Construct $\big\{{\D_k^{0}}^\I\big\}_{k=1}^K$\;
		Set $\I^{0} = \left(\T^{0}, \mathcal{P}^{0}\right)$\;
		\For{s in 1:S}{
			Sample $\boldsymbol \pi$, a permutation vector of length $K$\;
			Set $\{\D^s, \I^s\} = \{\D^{s-1}, \I^{s-1}\}$\;
			\For{$k$ in 1:K}{
				\If{$\boldsymbol{\pi}_k = 1$}{
					Construct $\cO_{\D^{s}}$ using Algorithm \ref{alg:validopD}\;
					Propose $\widetilde \D$ by sampling uniformly at random from $\cO_{\D^{s}}$\;
					Set $\D^s = \widetilde \D$ with probability $$\begin{aligned}
						\alpha_{\widetilde\D}  = 
						\textnormal{min}\Bigg\{1; &
						\frac{p\big(\bX\g\widetilde\D, \{I_s^{(j)}\}_{j \neq \boldsymbol{\pi}_k}\big)}
						{p\big(\bX\g\D^{s},\{I_s^{(j)}\}_{j \neq \boldsymbol{\pi}_k}\big)}
						\cdot
						\frac{p(\widetilde\D)}{p(\D^{s})}
						\cdot\frac{q(\D^{s}\g\widetilde   \D)}{q(\widetilde\D\g\D^{s})}\Bigg\}
					\end{aligned}$$\
				}
				\Else{
					Construct $\cO_{{\D_{\boldsymbol \pi_k}^{s\I}}}$ using Algorithm \ref{alg:validopDkI}\;
					Propose $\widetilde \D_{\boldsymbol{\pi}_k}^\I$ by sampling uniformly at random from $\cO_{{\D_{\boldsymbol{\pi}_k}^{s \I}}}$\;
					Recover $\widetilde I^{(\boldsymbol{\pi}_k)} = (\widetilde T^{({\boldsymbol{\pi}_k})}, \widetilde P^{({\boldsymbol{\pi}_k})})$ from  $(\widetilde \D_{\boldsymbol{\pi}_k}^\I, \D^s)$\;
					Set $I^{(\boldsymbol{\pi}_k)}_s =  \widetilde I^{(\boldsymbol{\pi}_k)}$ with probability
					$$\begin{aligned}
						\alpha_{\widetilde e_{\boldsymbol{\pi}_k}} = 
						\textnormal{min}\Bigg\{1; & 
						\frac{p\big(\bX\g \D^s, \{I_s^{(j)}\}_{j \neq \boldsymbol{\pi}_k}, \widetilde I^{(\boldsymbol \pi_{k})}\big)}
						{p\big(\bX\g\D^s,  \{I_s^{(j)}\}_{j \neq \boldsymbol \pi_k}, I_s^{(\boldsymbol \pi_{k})}\big)}
						\cdot
						\frac{p\big(\widetilde I^{(\boldsymbol{\pi}_k)}\big)}{p\big(I_s^{(\boldsymbol{\pi}_k)}\big)}
						\cdot\frac{q\big({\D_k^{s}}^\I\g\widetilde\D_{k}^\I\big)}{q\big(\widetilde\D_{k}^\I\g{\D_k^{s}}^\I\big)}\Bigg\}
					\end{aligned}$$\
				}
			}
		}
		Recover $\{\T^s, \mathcal{P}^s\}_{s=1}^S$ from $\{\I^s\}_{s=1}^S$\;
		\Return $\{\D^s, \T^s, \mathcal{P}^s\}_{s=1}^S$;
	}
	\caption{Random-scan MH to sample from $p(\D, \T, \mathcal{P} \g \bX)$}
	\label{alg:mhwithingibbs}
\end{algorithm}

Our main algorithm adopts the equivalent representation of $(\D,\I)$ in terms of $\I$-DAGs $\{\D_k^\I\}_{k=1}^K$. In this way, one can explore the space of possible pairs $(\D, \mathcal{I})$ using a set of simple operators inducing local modifications on DAGs. Specifically, we consider three types of operators: ${Insert}(u,v), {Delete}(u,v)$, and ${Reverse}(u,v)$, corresponding respectively to the insertion, deletion, and reversal of the edge $(u,v)$.
Also notice that the modified graph obtained by applying any of these operators may not be a DAG. Accordingly, we impose to the operators above the following \textit{validity} requirement (\textbf{vr}).
\begin{definition}
	Let $\{\D_k^\I\}_{k=1}^K$ be a sequence of $\I$-DAGs. An operator inducing a sequence of modified $\I$-DAGs $\{\widetilde \D_k^\I\}_{k=1}^K$ is \textnormal{valid} if every graph in $\{\widetilde \D_k^\I\}_{k=1}^K$ is a DAG.
\end{definition}
Let now $\mathcal{O}_{\D}$ be the set of all valid operators on DAG $\D$.
Our proposal distribution draws randomly an operator in $\mathcal{O}_{\D}$, and then apply it to $\D$ to obtain $\widetilde\D$. Accordingly, the (proposal) probability of a transition from $\D$ to $\widetilde\D$ is $q(\widetilde \D\g\D)=1/|\mathcal{O}_{\D}|$, where $|\mathcal{O}_{\D}|$ is the number of elements in $\mathcal{O}_{\D}$.
We use the same proposal scheme for the update of $\D_k^{\I}$.

Notice however that the same operator may imply different modifications when applied to the observational DAG $\D$ or to an $\I$-DAG $\D_k^{\I}$.
In the former case, the implied modification also affects all the $\I$-DAGs; in the latter case, the effect is local and affects only the $\I$-DAG corresponding to the $k$-th experimental setting. 
Accordingly, we need a different construction for the set of operators relative to the observational and experimental components.
Algorithm \ref{alg:validopD} constructs the set $\cO_\D$ simply by considering all possible valid insertions, deletions, and reversals of the edges of the observational DAG.
Differently, Algorithm \ref{alg:validopDkI} includes in $\cO_{\D_k^\I}$ all the operators implying: i) the insertion of an intervention target, ii) the modification of the parent set of a target node and iii) the deletion of an intervention target (provided that the parents of the target in the DAG and in the $\I$-DAG are the same).

\begin{algorithm}{
		\SetAlgoLined
		\vspace{0.1cm}
		\KwInput{A collection of $\I$-DAGs $\{\D_k^\I\}_{k=1}^K$}
		\KwOutput{A set of valid operators $\cO_\D$}
		Set $\cO_\D = \O$\;
		Construct $E_I = \{(u,v): \bA_{uv} = \bA_{vu} = 0\}$\;
		Construct $E_D = \{(u,v): \bA_{uv} = 1\}$\;
		\For{$e \in E_D$}{
			Add ${Delete}(e)$ to $\cO_\D$\;
			\lIf{${Reverse}(e)$ satisfies \textbf{vr}}{add it to $\cO_\D$}
		}
		\For{$e \in E_I$}{
			\lIf{${Insert}(e)$ satisfies \textbf{vr}}{add it to $\cO_\D$}
		}
		\Return $\cO_\D$;
	}
	\caption{Construction of $\cO_\D$}
	\label{alg:validopD}
\end{algorithm}

\begin{algorithm}{
		\SetAlgoLined
		\vspace{0.1cm}
		\KwInput{A collection of $\I$-DAGs $\{\D_k^\I\}_{k=1}^K$}
		\KwOutput{A set of valid operators $\cO_{\D_k^\I}$}
		Set $\cO_{\D_k^\I} = \O$\; 
		Recover $(T^{(k)}, P^{(k)})$ from $(\D, \D_k^\I)$\;
		\For{$v \notin T^{(k)}$}{
			Add ${Insert}(\zeta_k, v)$ to $\cO_{\D_k^\I}$\;
		}
		\For{$v \in T^{(k)}$}{
			\For{$u \in \text{nd}_{\D_k}(v)$}{
				\If{$u \in \pa_{\D_k}(v)$}{
					Add ${Delete}(u,v)$ to $\cO_{\D_k^\I}$\;
					\If{${Reverse}(u,v)$ satisfies \textbf{vr} and $u \in T^{(k)}$}{Add ${Reverse}(u,v)$ to $\cO_{\D_k^\I}$;}
				}
				\Else{
					Add ${Insert}(u,v)$ to $\cO_{\D_k^\I}$\;
				}
				\lIf{$\pa_{\D_k}(v) = \pa_{\D}(v)$}{add ${Delete}(\zeta_k, v)$ to $\cO_{\D_k^\I}$}
				
			}
		}
		\Return $\cO_{\D_k^\I}$;
	}
	\caption{Construction of $\cO_{\D_k^\I}$}
	\label{alg:validopDkI}
\end{algorithm}

The proposal distributions defined above are of key importance to ensure that the Markov chain implied by the Metropolis-Hastings is reversible, aperiodic and irreducible, so that the MCMC scheme provides an approximation of the posterior distribution, as stated in the following proposition.
\begin{proposition}
	\label{prop:markov:chain}
	The finite Markov chain defined by Algorithm \ref{alg:mhwithingibbs}, \ref{alg:validopD}, and \ref{alg:validopDkI} is reversible, aperiodic, and irreducible. Accordingly, it has $p(\D, \T, \mathcal{P} \g \bX)$ as its unique stationary distribution.
\end{proposition}



\subsection{Posterior Inference}
\label{sec:posterior:inference}

Output of Algorithm \ref{alg:mhwithingibbs} consists of a sample of size $S$ from the posterior distribution $p(\D, \T, \mathcal{P} \g \bX)$.
This MCMC output can be used to obtain summaries of specific features of the posterior distribution, such as DAG structures, both corresponding to the observational distribution of the variables, or a post-intervention distribution (represented by a modified DAG), as well as identifying the targets and parent sets induced by the interventions.

Point estimates of a DAG structure can be recovered through a Maximum A Posteriori (MAP) DAG estimate, corresponding to the DAG with the highest posterior probability, or based on the so-called Median Probability Model (MPM) originally introduced by
\citet{Barbieri:Berger:2004} in a linear regression setting. In this context, optimal properties of the MPM from a predictive viewpoint were also established by the authors.
To obtain an MPM-based estimate of a DAG we need to compute first a collection of marginal Posterior Probabilities of edge Inclusion (PPIs) for each possible directed link $(u,v)$ in any DAG $\D_k$.
Each corresponds to the $(u,v)$-element of a $(q,q)$ matrix $\boldsymbol{J}^{(k)}$,
\ben
\label{eq:post:probs:edge:inclusion}
\boldsymbol J_{uv}^{(k)} = \widehat p(u \to v \in \D_k \g \bX) = \frac{1}{S}\sum_{s=1}^S \mathbbm{1}\{u \to v \in \D_k^{s}\},
\een
where $\D_k^s$ is the modified DAG of context $k$ visited at iteration $s$.
When $k=1$ the above matrix collects the PPIs relative to $\D$, the DAG indexing the observational distribution.
An MPM DAG estimate, $\widehat\D_k$, for each $k\in[K]$, is finally obtained by including those edges whose PPIs is greater than $0.5$.

Now consider the intervention targets $T^{(1)},\dots,T^{(K)}$.
We can recover a marginal posterior probability of inclusion for a node $j\in[q]$ in the target $T^{(k)}$, $k \in \{2, \dots, K\}$, as
\ben
\label{eq:post:probs:targets}
\boldsymbol T^{(k)}_j = \widehat p (j \in T^{(k)}) = \frac{1}{S}\sum_{s=1}^S \mathbbm{1}\{j \in T^{(k)}_s\},
\een
while by definition $\boldsymbol T^{(1)}_j=0$ for each $j$.
The resulting collection of probabilities is organized in a $(q,K)$ matrix $\bT$
with $(k,j)$-element corresponding to $\boldsymbol{T}^{(k)}_j$.
As a point summary of the posterior distribution of $T^{(k)}$, we again consider a median-probability based estimate $\widehat{\boldsymbol{T}}^{(k)}$
such that, for each $j\in[q]$, $\widehat{\boldsymbol{T}}^{(k)}=1$ if $\boldsymbol{T}^{(k)}_j\ge 0.5$, $0$ otherwise.

A useful feature of our method is that it can be adopted to detect differences between experimental contexts that are reflected into modifications of the DAG structure, as induced by the interventions. These can be represented by means of a \emph{difference-graph} \citep{Wang:DCI:2018} which is constructed as follows.
Consider two DAGs $\D_1$ and $\D_k$, for $k\in\{2,\dots,K\}$. 
Let also $T^{(k)}$ be the intervention target associated with $\D_k$.
The difference-graph of $(\D_1,\D_k$), denoted as $\G^{(k)}$, is the graph whose adjacency matrix $\bG^{(k)}$ has $(u,v)$-element
\be
\bG^{(k)}_{uv}\,=\,
\left\{
\begin{array}{rl}
	1 & \textnormal{ if $v \in T^{(k)}$ and $u\in\{\pa_{\D_1}(v)\cup\pa_{\D_k}(v)\}$}, \\
	0 & \textnormal{ otherwise. }
\end{array}
\white\right\}
\ee
In other terms, an edge $u\to v$ is included in $\G^{(k)}$ whenever $v$ is an intervention target and $u$ is a parent of $v$ in at least on of the two DAGs, implying that the local distribution of node $v$ has been modified as the effect of a (soft or general) intervention.
For any $\G^{(k)}$ we can provide an MCMC-based estimate, $\widehat\G^{(k)}$ by following the same rationale leading to the MPM DAG and based on the collection of estimated PPIs.

\section{Simulations and Real Data Analysis}
\label{sec:5:sim:real}

In this section we apply our methodology for causal discovery under general interventions to simulated and real data.
To this end, in Section \ref{sec:gaussian} we first specialize our framework to Gaussian DAG models. 
In Section \ref{sec:simulations} we thus evaluate the performance of our method on simulated Gaussian data and compare it with alternative benchmark approaches. Finally, in Section \ref{sec:realdata} we present an application to biological protein expression data.

\subsection{Gaussian DAGs}
\label{sec:gaussian}

For the random vector $X=(X_1,\dots,X_q)^\top$, we consider a linear Gaussian Structural Equation Model (SEM) of the form
\begin{equation}
	\label{eq:sem:eqts}
	X = \bB^\top X + \bepsilon, \quad \bepsilon \sim \N_q(\bzero, \bD),
\end{equation}
where $\bB$ is a $(q,q)$ matrix of regression coefficients with $(l,j)$-element $\bB_{lj} \neq 0$ if and only if $l \in \pa_\D(j)$,
and $\bD=\textnormal{diag}(\bD_{11},\dots,\bD_{qq})$
is a $(q,q)$ matrix 
collecting the conditional variances of the $q$ variables.
Equivalently, we can write for each $j\in [q]$
\begin{align}
	\label{eq:pre:int:gaussian}
	X_j= \sum_{l \in \pa_{\D}(j)}\bB_{lj}X_l + \varepsilon_j, \quad \varepsilon_j \sim \N(0,\bD_{jj}).
\end{align}
Equation \eqref{eq:sem:eqts} implies
$X\g \bSigma, \D \sim \N_q( \bzero, \bSigma)$
with
$\bSigma = (\bI-\bB)^{-\top}\bD(\bI-\bB)^{-1}$, the right-hand side corresponding to the modified Cholesky decomposition of the covariance matrix.
Consider now a family of experimental settings with intervention targets $T^{(1)},\dots,T^{(K)}$ and implied modified DAGs $\D_1,\dots,\D_K$.
For each $k \in [K]$ we have
\begin{equation}
	\label{eq:post:int:gaussian}
	X_j = \sum_{l \in \pa_{\D_k}(j)} {\bB}^{(k)}_{lj}X_l + \varepsilon^{(k)}_j, \quad \varepsilon^{(k)}_j \sim \N\big(0,\bD^{(k)}_{jj}\big), \quad j \in T^{(k)},
\end{equation}
where $(\bB^{(k)}, \bD^{(k)})$ are the DAG-parameters induced by the general intervention. Notice that all the $(l,j)$-elements of $(\bB^{(k)}, \bD^{(k)})$ not involved in \eqref{eq:post:int:gaussian} are exactly those in $(\bB, \bD)$ because of the assumed invariances between pre- and post-intervention distributions (see Equations \eqref{eq:intfactor} and \eqref{eq:likelihoodfinal}).
For each experimental setting $k \in [K]$, the post-intervention joint distribution of $X$ is then
	$
	X\g \bSigma_k, \D_k \sim \N_q\big( \bzero, \bSigma_k\big),
	$
where $\bSigma_k = \big(\bI-\bB^{(k)}\big)^{-\top}\bD^{(k)}\big(\bI-\bB^{(k)}\big)^{-1}$.
Because of the prior elicitation procedure introduced in Section \ref{sec:causal:discovery:general:interventions}, to compute the DAG marginal likelihood \eqref{eq:marg:like:general} we only need to specify a prior for the parameter of a complete (unconstrained) Gaussian DAG model.
It is immediate to show that assumptions \textbf{A1-A3} of Section \ref{sec:parameter:priors:elicitation} are satisfied in the Gaussian setting by $\bOmega\sim\W_q(a, \bU)$, namely a Wishart distribution on $\bOmega=\bSigma^{-1}$ having expectation $a\bU^{-1}$ with $a>q-1$ and $\bU$ a $(q,q)$ s.p.d.~matrix.
By combining such prior with the likelihood of $n$ i.i.d.~samples from $\N_q(\bzero,\bSigma)$, we obtain the following formula for the marginal data distribution relative to any subset of the $q$ variables $B \subset [q]$:
\begin{equation}
	\label{eq:gauss:marklik:B}
	p(\bX_{.B}) = \pi^{-\frac{n|B|}{2}}\frac{|\bU_{BB}|^{\frac{a-|\bar B|}{2}}}{|\widetilde \bU_{BB}|^{\frac{a-|\bar B|+n}{2}}}\frac{\Gamma_{|B|}\left(\frac{a-|\bar B| + n}{2}\right)}{\Gamma_{|B|}\left(\frac{a-|\bar B|}{2}\right)},
\end{equation}
where $\bar B = [q] \backslash B$ and $\widetilde \bU = \bU + \bX^{\top}\bX$; see for instance \citet{Press:1982}.
This formula, implemented in Equation \eqref{eq:marg:like:general} for suitable elements (rows and columns) of the data matrix $\bX=\big(\bX^{(1)},\dots,\bX^{(K)}\big)^{\top}$,
specializes the DAG marginal likelihood to the Gaussian setting.
Note that
the resulting marginal likelihood
provides an adaptation to our interventional setting
of the popular Bayesian Gaussian equivalent (BGe) score, originally introduced by \citet{Heckerman:Geiger:1995:BDeu:BGe} for the case of i.i.d.~observational data; see also \citet{Geiger:Heckerman:2002}.
When coupled with the model prior introduced in Section \ref{sec:priors:model}, this result fully specializes our general methodology to the Gaussian setting.

\subsection{Simulation Studies}
\label{sec:simulations}

We evaluate the performance of our method under several simulated scenarios where we vary
i) the number of experimental settings $K \in \{2,4\}$, ii) the number of variables $q \in \{10,20\}$ and iii) the sample size $n_k \in \{100, 500, 1000\}$ that we assume equal across $k\in[K]$.

For each combination of $K$ and $q$, $40$ true DAGs, intervention targets and induced parent sets are generated as follows.
We first draw a sparse DAG $\D$ with a probability of edge inclusion $3/(2q-2)$, so that the expected number of edges in the DAG grows linearly with the number of variables  \citep{peters:buehlmann:2014}.
Each target $T^{(k)}$, $k \in \{2, \dots, K\}$, is then generated by randomly including each node $j \in [q]$ in $T^{(k)}$ with probability $\eta_k=0.2$.
For each node $j \in T^{(k)}$, consider now matrix $\bP^{(k)}$ which represents the (possibly different) parent sets induced by the intervention; the latter is constructed by randomly generating a new DAG with same topological ordering as $\D$, and replacing the original parent set of $j$ with that of the new DAG.
Finally, conditionally on DAG $\D$ and the so-obtained modified DAGs $\D_2,\dots,\D_K$,
we draw the set of distinct parameters $\bB_{lj}^{(k)}$ uniformly in $[-1,-0.1]\cup[0.1,1]$, while we fix $\bD^{(k)}_{jj}=1$ for each $j\in[q]$ and $k\in[K]$.
Finally, by recovering $\bSigma_k$ from $\big(\bB^{(k)},\bD^{(k)}\big)$, $n_k$ observations are generated
from $\N_q\big( \bzero, \bSigma_k\big)$,
for $k\in[K]$.
Output is finally a collection of simulated datasets $\bX^{(1)},\dots,\bX^{(K)}$.

We implement our method by running Algorithm \ref{alg:mhwithingibbs} for number of MCMC iterations $S=3000q$, discarding the initial $1000q$ draws that are used as a burn-in period.
We set $a_{\phi}=b_{\phi}=1$, $a_{\eta}=b_{\eta}=1$ and $a_{\D}=a_{\D}=1$ in the hierarchical model priors of Section \ref{sec:priors:model}. These specific choices result in uniform priors for the inclusion of a node in an intervention target \eqref{eq:prior:target}, as a new parent \eqref{eq:prior:parent:sets} as well as for the probability of edge inclusion in $\D$ \eqref{eq:prior:dag}.
Finally, we set $a=q$ and $\bU=\bI_q$ in the Wishart prior on $\bOmega$, leading to a weakly informative prior whose weight corresponds to a sample of size one.

We evaluate the performance of our method in the tasks of DAG learning and target identification.
To this end, we consider as point estimates of DAGs and targets
the Median Probability DAG model and Median Probability Targets as introduced in Section \ref{sec:posterior:inference}.
Since there are no existing methods for causal discovery that
align precisely with our framework of general interventions,
providing a fully equitable comparison is not straightforward.
To address this issue, we benchmark our approach against alternative methodologies designed for slightly different contexts. Specifically, we consider three methods: GIES \citep{Hauser:Buehlmann:2012}, its recent extension GnIES \citep{gamella2022characterization}, and UT-IGSP \citep{squires:utigsp:2020}.

GIES, which requires exact knowledge of the intervention targets, serves as a reference for the DAG structure learning task. In contrast, both GnIES and UT-IGSP
learn the intervention targets from the data, but assume slightly different definitions of interventions.
Specifically, GnIES
considers \emph{noise-interventions}, which only modify the error-term distribution of the interventioned nodes in \eqref{eq:scmdef}.
Differently, UT-IGSP works under the framework of \textit{soft interventions}.

Although the interventions considered by the methods above produce different post-intervention distributions, the implied invariances coincide, thus making our comparison sensible.
In addition, all benchmarks provide an I-Essential Graph (I-EG) estimate which represents an I-Markov equivalence class of DAGs.
We therefore adapt the MPM DAG estimate provided by our method by constructing the representative I-EG.
Figure \ref{fig:SHD:cfr} summarizes the Structural Hamming Distance (SHD) between each I-EG estimate and true I-EG, for all methods under comparison; SHD is defined as the number of insertions, deletions or flips needed to transform the estimated graph into the true DAG; accordingly lower values of SHD imply better performances.

\begin{figure}
	\begin{center}
		\begin{tabular}{lcc}
			& $\quad\quad  K=2$ & $\quad\quad  K=4$ \\
			\multirow{1}{*}{\rotatebox[origin=c]{90}{$q=10$\quad}} &
			\raisebox{-0.7\height}{\includegraphics[width=0.45\linewidth]{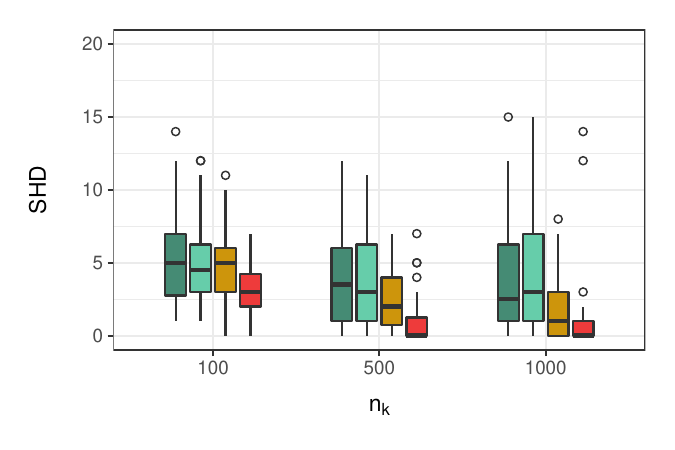}} &
			\raisebox{-0.7\height}{\includegraphics[width=0.45\linewidth]{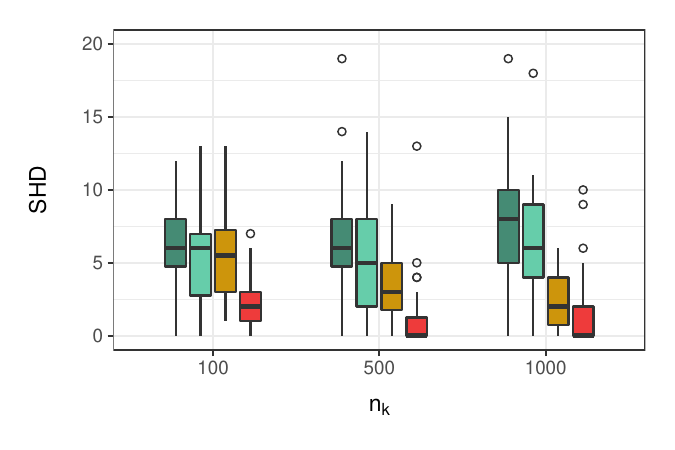}} \\
			\multirow{1}{*}{\rotatebox[origin=c]{90}{$q=20$\quad}} &
			\raisebox{-0.7\height}{\includegraphics[width=0.45\linewidth]{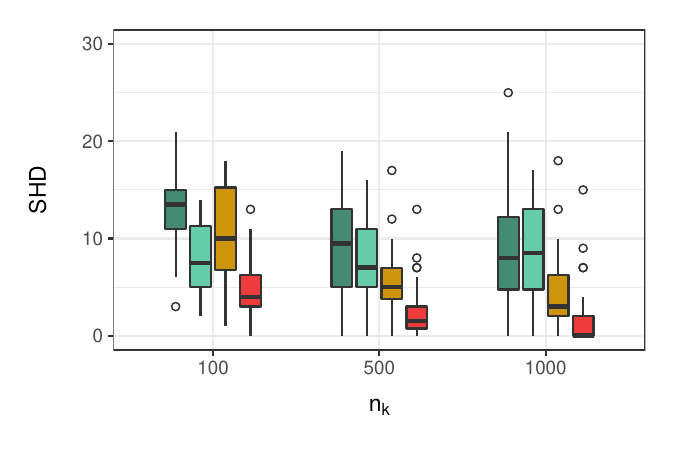}} &
			\raisebox{-0.7\height}{\includegraphics[width=0.45\linewidth]{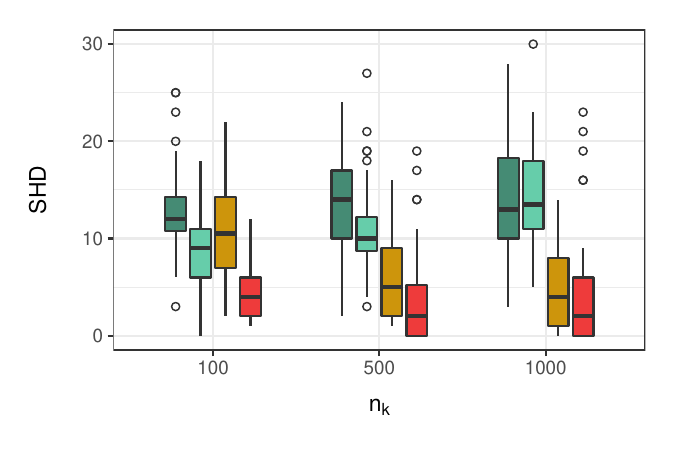}}
		\end{tabular}
		\caption{Simulations. Distribution (across $40$ simulations) of the Structural Hamming Distance (SHD)        between true DAG and graph estimate, under scenarios $q\in\{10,20\}$ (number of variables), $K\in\{2,4\}$ (number of experimental contexts), and for increasing samples sizes $n_k\in\{100,500,1000\}$. Methods under comparison are: GIES and GnIES (dark and light blue), UT-IGSP (yellow) and our Bayesian approach (red).}
		\label{fig:SHD:cfr}
	\end{center}
\end{figure}

Figure \ref{fig:targets:errors:cfr} instead reports the number of errors (both false positives and false negatives) relative to target identification for our method, GnIES and UT-IGSP.
Our method exhibits a superior performance in comparison with the benchmarks, as also expected because of deviations of the simulated data
from the assumptions underlying their methods.
Therefore, the two benchmarks reveal difficulties in recovering a causal DAG structure from interventional data whose generating mechanism is consistent with a broader, namely \textit{general}, framework of interventions.

\begin{figure}
	\begin{center}
		\begin{tabular}{lcc}
			& $\quad \quad K=2$ & $\quad \quad K=4$ \\
			\multirow{1}{*}{\rotatebox[origin=c]{90}{$q=10$\quad}} &
			\raisebox{-0.7\height}{\includegraphics[width=0.45\linewidth]{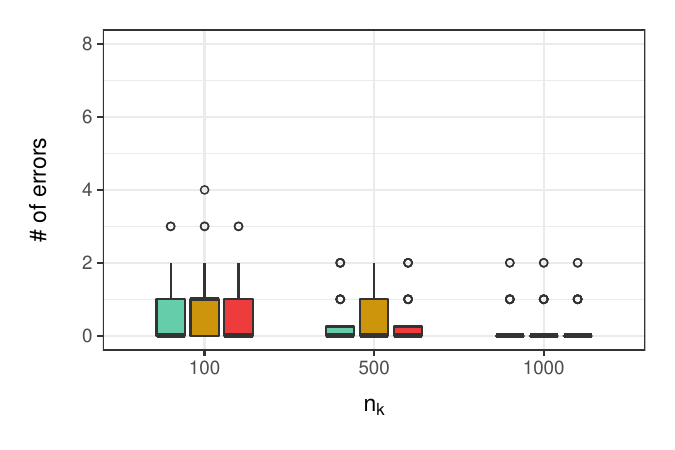}} &
			\raisebox{-0.7\height}{\includegraphics[width=0.45\linewidth]{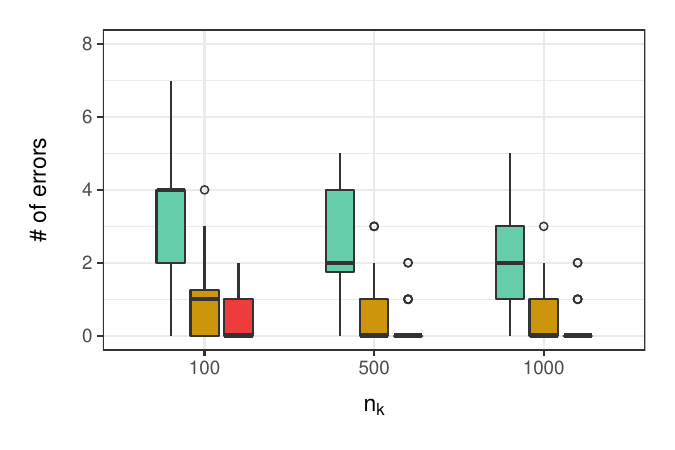}} \\
			\multirow{1}{*}{\rotatebox[origin=c]{90}{$q=20$\quad}} &
			\raisebox{-0.7\height}{\includegraphics[width=0.45\linewidth]{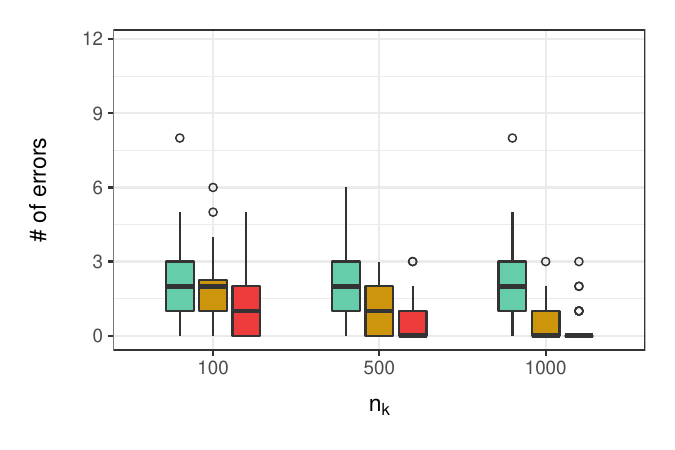}} &
			\raisebox{-0.7\height}{\includegraphics[width=0.45\linewidth]{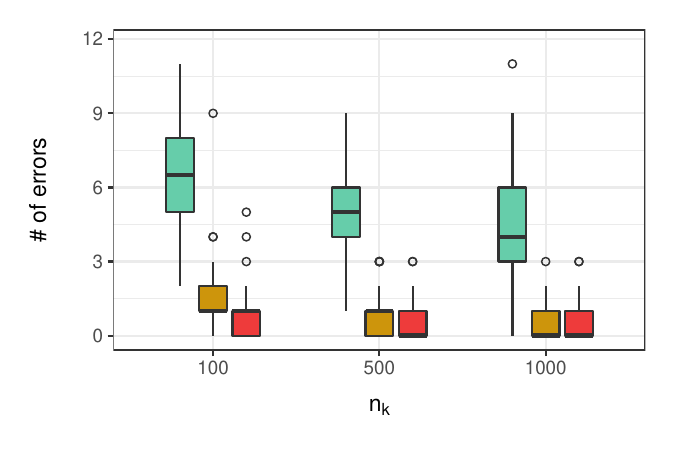}}
		\end{tabular}
		\caption{Simulations. Distribution (across $40$ simulations) of the number of false positives and false negatives (\# of errors) between true and estimated targets, under scenarios $q\in\{10,20\}$ (number of variables), $K\in\{2,4\}$ (number of experimental contexts), and for increasing samples sizes $n_k\in\{100,500,1000\}$. Methods under comparison are: GnIES (light blue), UT-IGSP (yellow) and our Bayesian approach (red).}
		\label{fig:targets:errors:cfr}
	\end{center}
\end{figure}

As described in Section \ref{sec:posterior:inference}, the output provided by our method can be also adapted to learn differences between DAGs corresponding to different experimental settings.
For this specific goal, \cite{Wang:DCI:2018} developed the Difference Causal Inference (DCI) algorithm.
To assess the performance of our method in this context and compare it with DCI, we consider the same simulation scenarios for $K=2$ defined before.
With regard to DCI, we consider two implementations.
In the first one, following \citet{Belyaeva:et:al:2020:DCI}, we set $\alpha_{ug} = 0.001$, $\alpha_{sk} = 0.5$ and $\alpha_{dd} = 0.001$ as confidence levels for the tests used in the corresponding three steps of the algorithm. In the second one, we implement DCI with stability selection with input the grid of possible hyperparameters defined by $\alpha_{ug} \in \{0.001, 0.01\}$, $\alpha_{sk} \in \{0.1, 0.5\}$ and $\alpha_{dd} \in \{0.001, 0.01\}$. 
\black
Figure \ref{fig:DCI:cfr} summarizes the sum of falsely identified and non-identified
edges in the estimated difference-graph of $(\D_1,\D_2)$.
Both methods improve their ability in recovering structural differences between the two DAGs as the sample size increases.
Moreover, the performance of our method is slightly better than DCI, expecially under the $q=20$ scenario.

\begin{figure}
	\begin{center}
		\begin{tabular}{cc}
			$\quad\quad\quad q=10$ & $\quad\quad\quad q=20$ \\
			\includegraphics[width=0.48\linewidth]{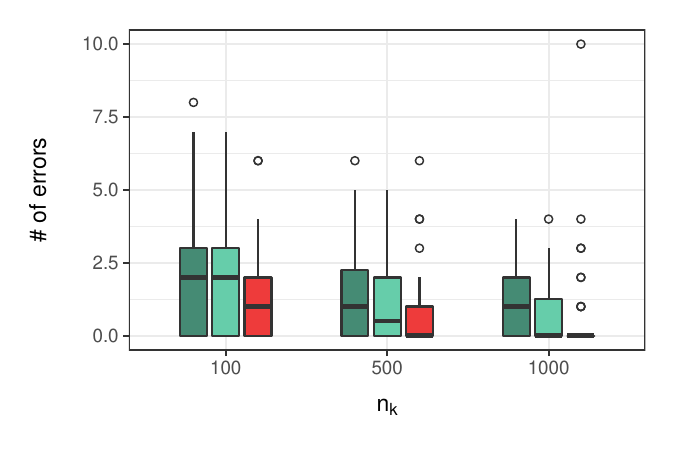} &
			\includegraphics[width=0.48\linewidth]{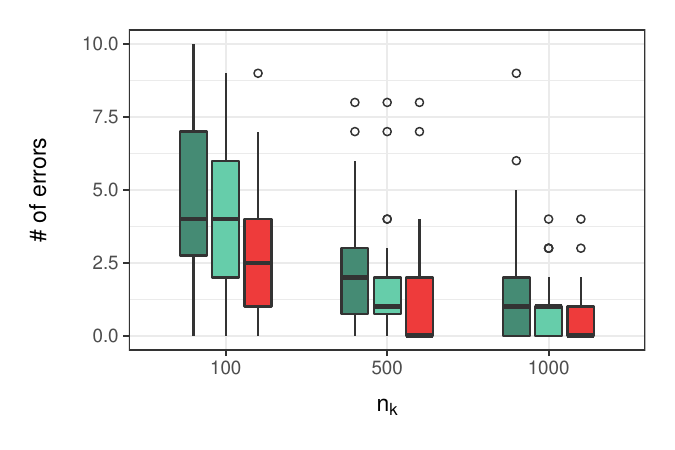}
		\end{tabular}
		\caption{Simulations. Distribution (across $40$ simulations) of the sum of falsely identified and non-identified varying edges between context $k=1$ and $k=2$, under scenarios $q\in\{10,20\}$ (number of variables) and for increasing samples sizes $n_k\in\{100,500,1000\}$. Methods under comparison are: DCI and DCI with stability selection (dark and light blue) and our Bayesian approach (red).}
		\label{fig:DCI:cfr}
	\end{center}
\end{figure}

\subsection{Real data analysis}
\label{sec:realdata}

We apply our methodology to a dataset of protein expression measurements from patients affected by Acute Myeloid Leukemia (AML). Subjects are classified into groups corresponding to distinct AML subtypes which were identified according to the French-American-British (FAB) system based on morphological features, cytogenetics, and assessment of recurrent molecular abnormalities.
The complete dataset is provided as a supplement to \cite{Kornblau:et:al:2009} and was previously analyzed from a multiple graphical modelling perspective by \cite{Peterson:et:al:2015:JASA} and \cite{Castelletti:et:al:2020:SIM}. Specifically, the authors developed Bayesian methodologies
to infer a distinct graphical structure for each group (subtype), and simultaneously allowing for similar features across groups through a hierarchical prior on graphs favoring network relatedness.
Given the distinct prognosis associated with each AML subtype, it is reasonable to expect variations in protein interactions among groups, as revealed by the analysis of \citet{Castelletti:et:al:2020:SIM}. The investigation of such variations is of great interest from a therapeutic perspective, since it can provide valuable insights on the efficacy of a treatment capable of protein regulation depending on the specific patient's subtype; see also \citet{Castelletti:Consonni:2023:SIM}.

Similarly to \citet{Peterson:et:al:2015:JASA}, we consider the level of $q = 18$ proteins and phosphoproteins involved in apoptosis and cell cycle regulation according to the KEGG database, relative to $n=178$ diagnosed AML patients corresponding to the following $K = 4$ subtypes: M0 (17 subjects), M1 (34 subjects), M2 (68 subjects) and M4 (59 subjects). We designate the largest group, M2, as the observational reference group, and attribute differences among subtypes to unspecified general interventions that may have altered the reference network structure.
%
We implement our methodology by running Algorithm \ref{alg:mhwithingibbs} for a number of MCMC iterations $S=250000$, and discarding the initial $50000$ draws which are used as a burn-in period.
We consider for all priors the same weakly informative hyperparameter choices employed in the simulation study of Section \ref{sec:simulations}.

As a summary of the MCMC output we first compute the marginal posterior probability of target inclusion according to Equation \eqref{eq:post:probs:targets} for each node $v\in[q]$ and AML subtype (experimental context $k$).
The resulting collection of probabilities is summarized in the heat map of Figure \ref{fig:PPIs:targets}.
Results show that a few proteins
are with high probability targeted as the result of unknown interventions that affect the network of protein interactions under any of the subtypes.
Specifically, only four proteins, namely BCL2 and CCND1 under Subtype M1 and
GSK3 and XIAP under Subtype M4, are identified as intervention targets with a posterior probability exceeding 0.5. Differences in the implied set of parent-child relations involving such nodes are therefore expected in the implied post-intervention graphs.
By converse, there are no proteins whose probabilities of intervention are higher than the 0.5 threshold under Subtype M0.




\begin{figure}
	\begin{center}
		\includegraphics[width=0.7\linewidth]{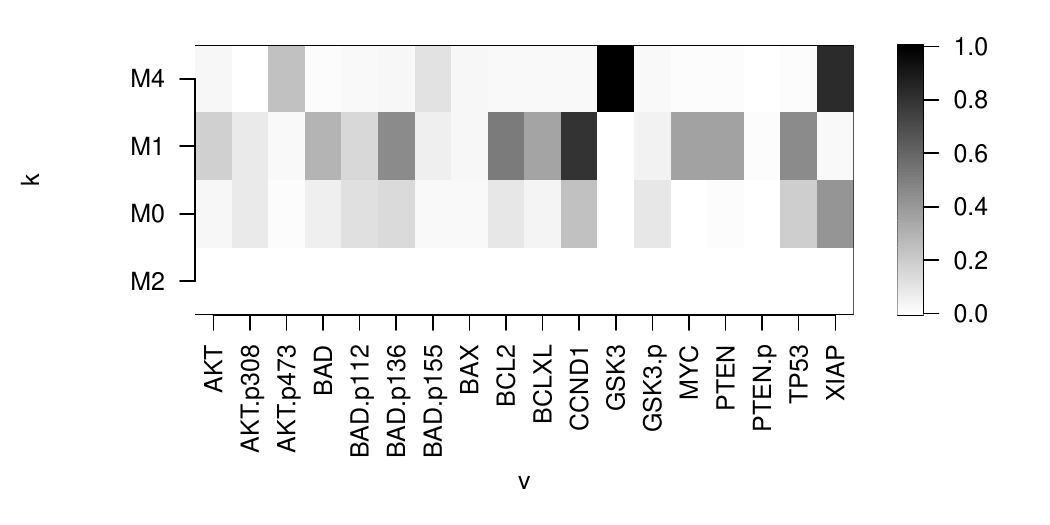}
		\caption{AML data. Estimated marginal posterior probabilities of target inclusion, computed for each node $v \in [q]$ across AML subtypes, each corresponding to an experimental context $k$. Subtype M2 corresponds to the reference (observational) context.}
		\label{fig:PPIs:targets}
	\end{center}
\end{figure}

According to Equation \eqref{eq:post:probs:edge:inclusion}, we then compute the Posterior Probability of Inclusion (PPI) for each possible directed edge $(u,v)$ and each group-specific post-intervention DAG, corresponding to one of the four subtypes. Results for each subtype M0, M1, M2, M4 are reported in the $(q,q)$ heat maps of Figure \ref{fig:PPIs:edge:inclusion}, where any $(u,v)$-element in the plots corresponds to the marginal probability of inclusion of $u\rightarrow v$ in one of the four DAGs.

Finally, as single graphs summarizing the entire MCMC output, we provide
a collection of context-specific MPM DAG estimates, $\widehat\D_k,k=1,\dots,4$.
These are reported in Figure \ref{fig:DAGs:leukemia}, where for ease of interpretation the graph indexing the observational context (Subtype M2) corresponds to the I-EG representing the equivalence class of the estimated DAG.
As expected from the previous results, the four graphs exhibit several similarities. An instance is the path involving the PTEN, PTEN.p and BAD.p136, BAD.p155 proteins. Such associations
are consistent with findings in \citet{Peterson:et:al:2015:JASA} who also identified (undirected) links between these proteins under all groups.
In addition, our method detects a direct effect of BAD.p136 on PTEN.p, as well as of PTEN on BAD.p155 for all leukemia patients.
A notable difference across groups is instead represented by the absence of the directed link AKT $\rightarrow$ GSK3 in group M4 as the effect of a (hard) intervention targeting GSK3 and which removes its parents.
Notably, the correlation of GSK3 with a number of proteins involved in AML, and primarly AKT, was established in the medical literature; see for instance \citet{Ruvolo:et:al:2015} and \citet{Ricciardi:et:al:2017}. In particular, the AKT/GSK3 path was shown to represent a critical axis in AML, which may be a therapeutic target in AML patients with intermediate cytogenetics (M2 subtype).
Out results show that an \textit{intervention} on AKT aimed at regulating the GSK3 protein may be beneficial for patients characterized by AML subtypes M0, M1, M2, while uneffective whenever applied to M4 patients since there are no paths from AKT downstreaming to GSK3.

\begin{figure}
	\begin{center}
		\begin{tabular}{cc}
			Subtype M2 & Subtype M0 \\
			\includegraphics[width=0.48\linewidth]{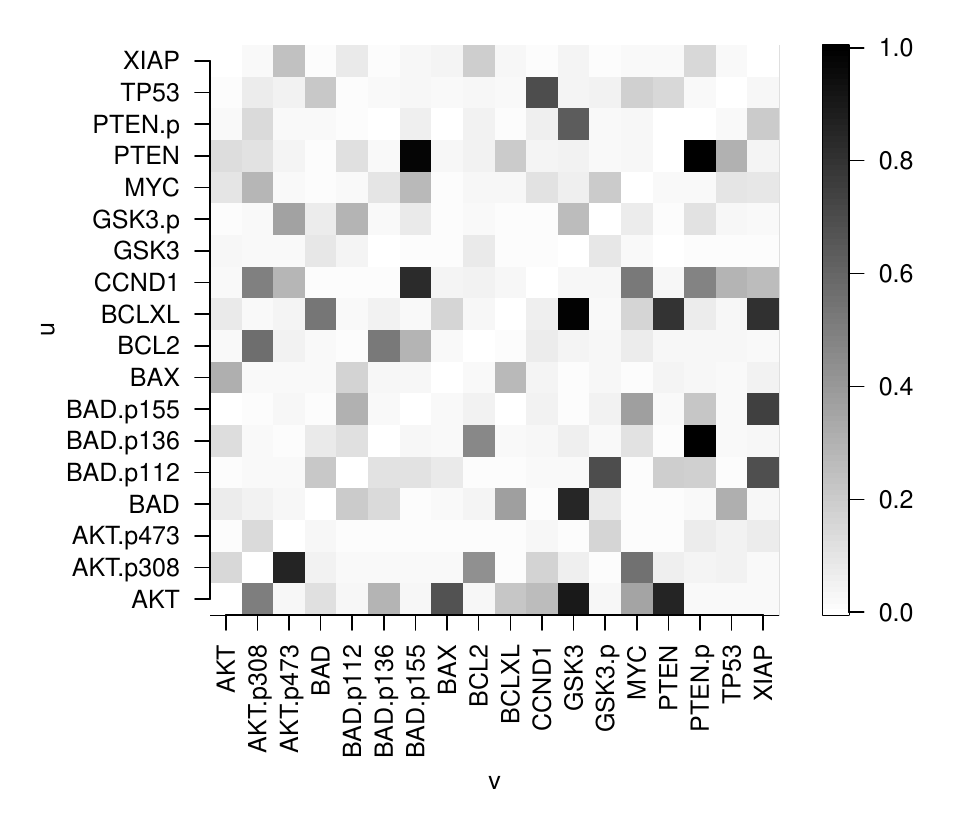} &
			\includegraphics[width=0.48\linewidth]{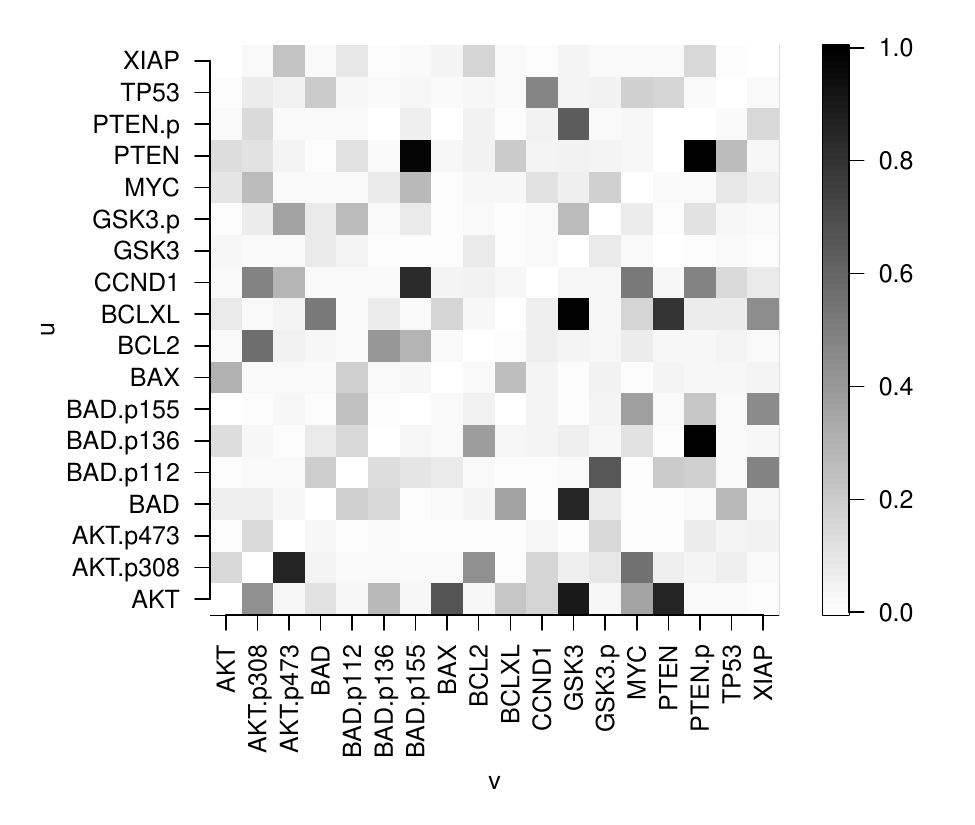} \\
			Subtype M1 & Subtype M4 \\
			\includegraphics[width=0.48\linewidth]{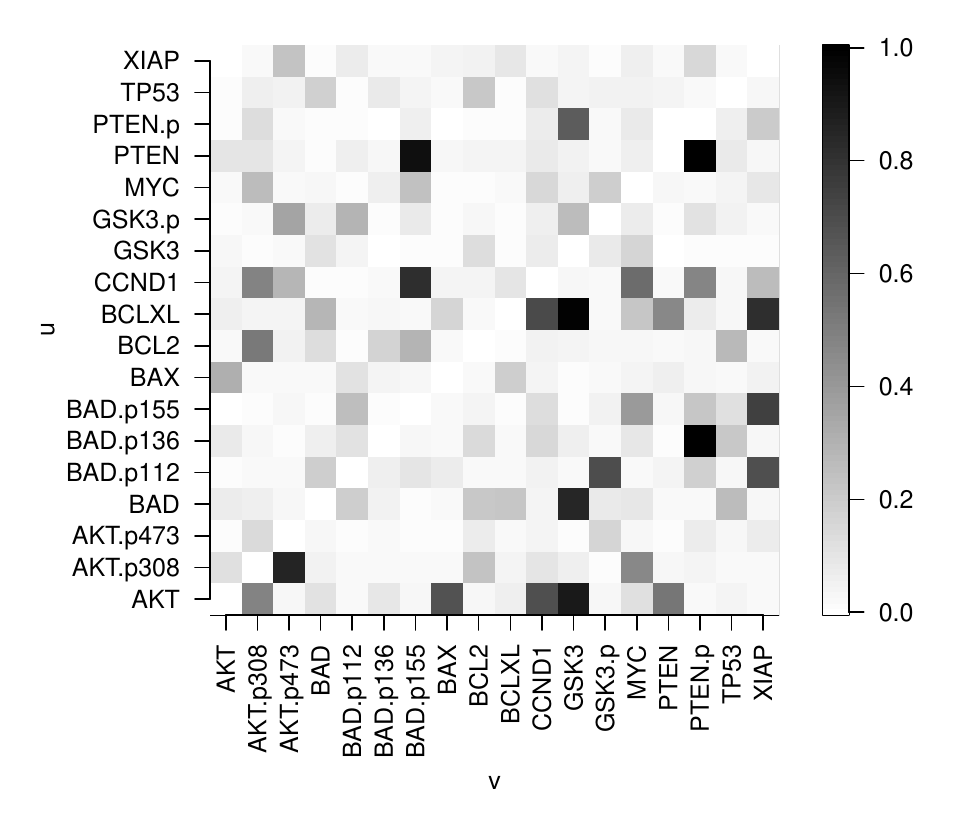} &
			\includegraphics[width=0.48\linewidth]{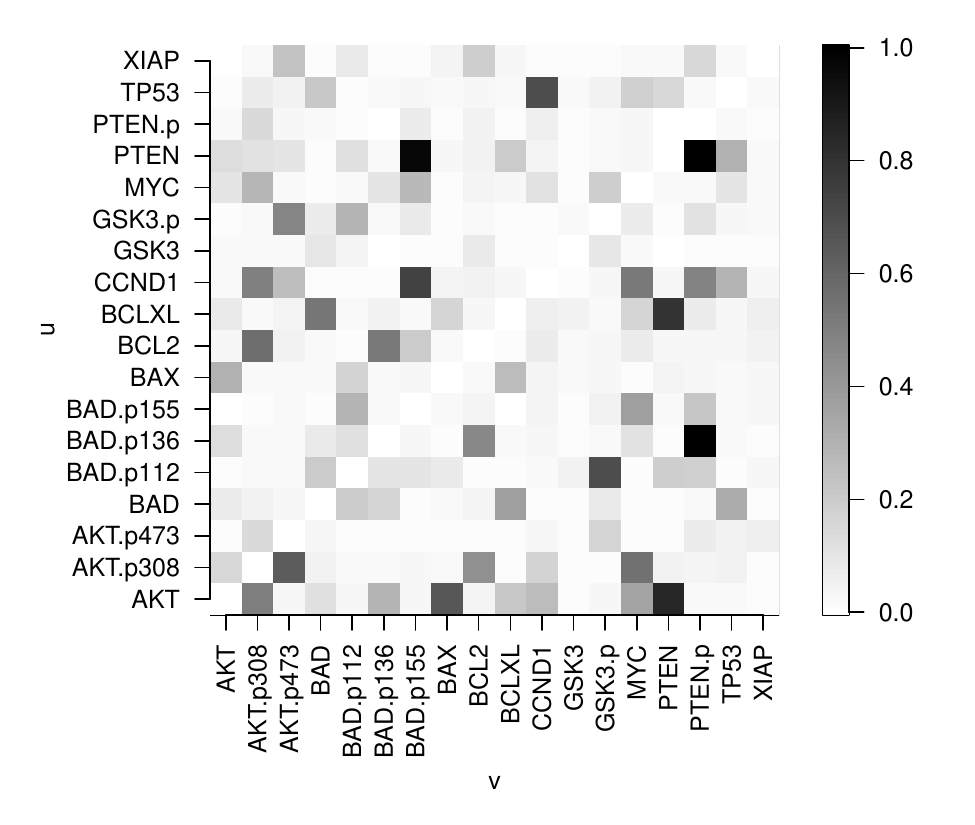}
		\end{tabular}
		\caption{AML data. Estimated marginal posterior probabilities of edge inclusion, computed for each possible directed edge $(u,v)$, $u,v \in [q]$ and group-specific post-intervention DAG, each corresponding to one of the four AML subtypes.}
		\label{fig:PPIs:edge:inclusion}
	\end{center}
\end{figure}

\begin{figure}
	\begin{center}
		\begin{tabular}{cc}
			Subtype M2 & Subtype M0 \\
			\includegraphics[width=0.48\linewidth]{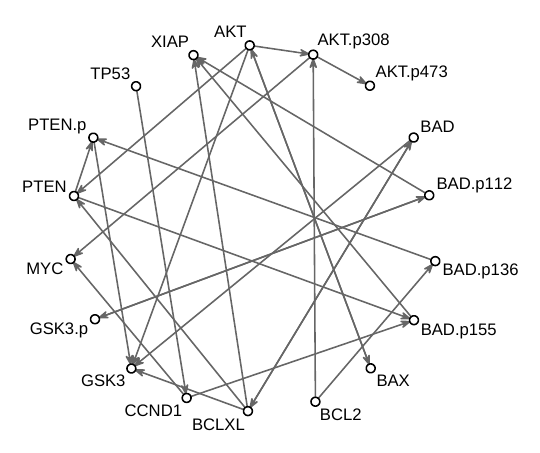} &
			\includegraphics[width=0.48\linewidth]{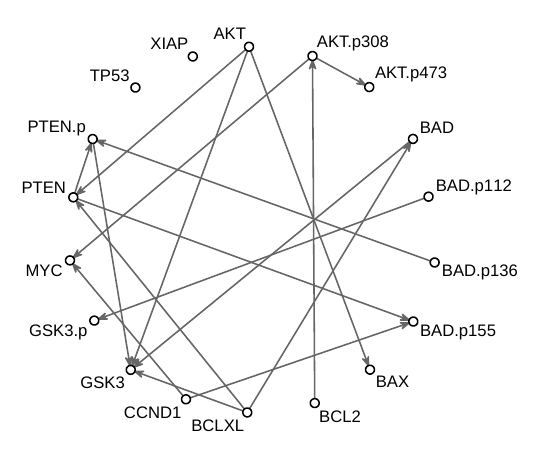} \\
			& \\
			Subtype M1 & Subtype M4 \\
			\includegraphics[width=0.48\linewidth]{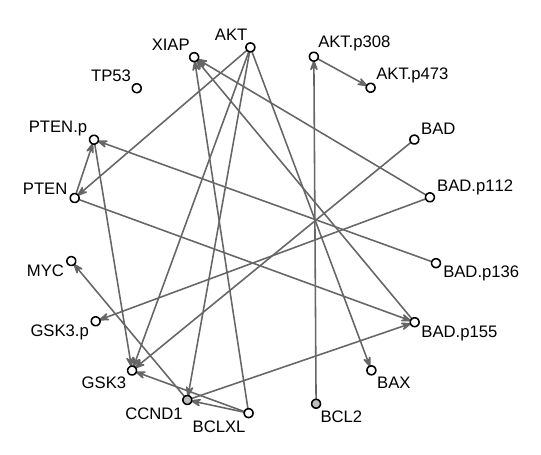} &
			\includegraphics[width=0.48\linewidth]{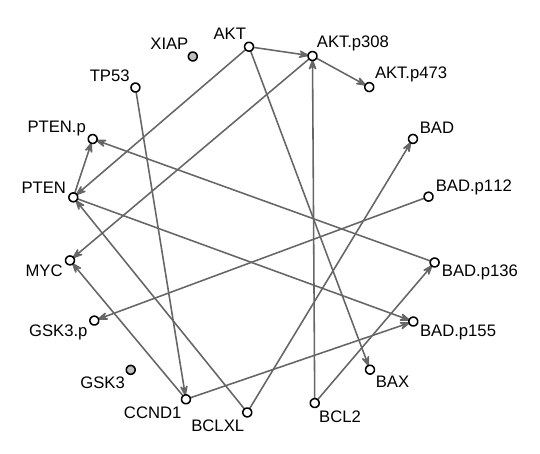}
		\end{tabular}
		\caption{AML data. Median Probability graph Model (MPM) estimates obtained under each AML subtype. Graph corresponding to Subtype M2 is the representative I-EG.}
		\label{fig:DAGs:leukemia}
	\end{center}
\end{figure}

\section{Discussion}
\label{sec:discussion}

In this paper we introduce a statistical framework for causal discovery from multivariate interventional data.
The notion of general intervention that we implement allows for
structural modifications in the parent-child relations involving the intervened nodes, where the latter can be both known in advance or completely uncertain.
Under both contexts, we first establish DAG identifiability
and provide graphical criteria to characterize interventional Markov equivalence of DAGs.
We then develop a Bayesian methodology for structure learning, by introducing an effective procedure which dramatically simplifies parameter prior elicitation. In addition, it provides a closed-form expression for the DAG marginal likelihood which
guarantees score equivalence among I-Markov equivalent DAGs.
We complete our Bayesian model formulation by assigning priors to model parameters corresponding to DAGs, intervention targets, and modified parent sets.
Finally, to approximate the corresponding posterior distribution, we develop a Markov Chain Monte Carlo (MCMC) sampler based on a random scan Metropolis Hastings scheme.

\subsection{Future Developments}

Our Bayesian framework for causal discovery relies on a set of general assumptions on the likelihood and prior that are satisfied under various parametric families, and notably zero-mean Gaussian models, when equipped with a Wishart prior on the precision matrix.
Within such context, the full development of a methodology for structure learning and target identification is possible, and asymptotic properties relative to posterior ratio consistency could be established along the lines of
\citet{Castelletti:Peluso:Bka:2023} and \citet{Castelletti:Peluso:2022}
for the case of known and unknown hard interventions respectively.
Similarly, our framework can be implemented for the analysis of categorical DAGs, under a multinomial-Dirichlet model. The resulting method would extend the original methodology of \citet{Heckerman:et:al:1995:BDeu}, developed for i.i.d.~observational samples and leading to their BDeu score, to an experimental setting of general (unknown) interventions.

Our approach for causal discovery is based
on the assumption that the data are generated according to a Markovian Structural Causal Model (SCM) with no cycles, and which can be thus represented by a directed \emph{acyclic} graph.
Besides the absence of cycles, our SCM representation assumes that there are no latent (unmeasured) confounders.
Recently,
\citet{Bongers:et:al:AOS:2021:cycles}
proposed a general theory for causal discovery which allows for the presence of both latent confounders and cycles,
establishing identifiability conditions of SCMs as well as several statistical properties of their methodology.
An extension of our method for causal discovery under general interventions towards this direction can be also of interest.






\newpage

\section*{Appendix A. Proofs of Section 2}
\label{appx:proofs}

This section contains all the proofs of the main results presented in Sections 2.2 and 2.3 of the paper.
Numbering of propositions and theorems in this section is the same as in the main text. Auxiliary lemmas and propositions that are newly introduced within this appendix follow instead the sequential numbering in line with the main text.



\subsection*{A.1 Proofs of Section 2.2}
\label{appx:proofs2:2}

The I-Markov property of Definition \ref{def:imarkovprop} and the graphical characterization of I-Markov equivalence of Theorem \ref{theorem:imarkeq:skeleta} is similar to the one provided by \citet{Yang:IGSP:2018} for the case of soft interventions. As a consequence, our proofs of Proposition \ref{prop:imarkpropok} and Theorem \ref{theorem:imarkeq:skeleta} are adapted from the ones of Proposition 3.8 and Theorem 3.9 in their paper and are here reported for completeness.

We first characterize I-Markov equivalence in our setting in terms of the ensued factorization:

\begin{lemma}
	\label{lemma:factorization}
	$\{p_k(\cdot)\}_{k=1}^K \in \M_{\I}(\D)$ if and only if there exists $p(\cdot) \in \M(\D)$ such that, for each $k \in [K]$, $p_k(\cdot)$ factorizes as $\prod_{j \notin T^{(k)}} p(x_{j}\g \bx_{\pa_\D(j)}) \prod_{j \in T^{(k)}} p_k(x_{j}\g \bx_{{\pa}_{\D_k}(j)})$.
\end{lemma}

\begin{proof}
	\textit{If} - Suppose there exists $p(\cdot) \in \M(\D)$ such that the factorization above holds. The first condition from the definition of the I-Markov equivalence class, namely that $p_k(\bx) \in \M(\D_k)$ is trivially satisfied for all $k \in [K]$. As for the second condition, note that for all $j \notin T^{(k)}$ we have $p_k(x_j\g\bx_{\pa_{\D_k}(j)}) = p_k(x_j\g \bx_{\pa_\D(j)}) = p(x_j\g \bx_{\pa_\D(j)})$. As a consequence, $p_k(x_j\g \bx_{\pa_{\D_k}(j)}) = p(x_j\g \bx_{\pa_\D(j)}) = p_{k^\prime}(x_j\g \bx_{\pa_{\D_{k^\prime}}(j)})$, $\forall \, j \notin T^{(k)} \cup T^{(k^\prime)}$ and $T^{(k)}, T^{(k^\prime)} \in \T$. Hence $\{p_k(\bx)\}_{k = 1}^K \in \M_{\I}(\D)$.
	
	\textit{Only if} - Suppose that $\{p_k(\bx)\}_{k=1}^K \in \M_{\I}(\D)$. To prove that there exists $p(\bx) \in \M(\D)$ such that the factorization in the lemma holds, take any $p(\bx) \in \M(\D)$. By definition, it holds that $p_k(\bx) = \prod_{j=1}^q p_k(x_j\g\bx_{\pa_{\D_k}(j)})$. From the second condition, we have that for any $k \in [K]$ and $j \notin T^{(k)}$, $p_k(x_j\g\bx_{\pa_{\D_k}(j)}) = p(x_j\g\bx_{\pa_\D(j)})$, where $p(x_j\g\bx_{\pa_\D(j)})$ is an arbitrary strictly positive density, so that the factorization in the lemma holds for all $T \in \T$.
\end{proof}

\begin{repproposition}{prop:imarkpropok}
	Let $\D$ be a DAG and $\I$ a collection of targets and induced parent sets. Then $\{p_k(\cdot)\}_{k=1}^K \in \M_{\I}(\D)$ if and only if $\{p_k(\cdot)\}_{k=1}^K$ satisfies the I-Markov property with respect to $\{\D^\I_k\}_{k=1}^K$.
\end{repproposition}

\begin{proof}
	\textit{If} - Choose any $k \in [K]$ and use the chain rule to factorize $p_k(\cdot)$ according to the topological ordering of $\D_k$, so that
	$$p_k(\bx) = \prod_{j=1}^q p_k(x_j\g\bx_{a_j(\pi_{\D_k})})$$
	where $a_j(\pi_{\D_k})$ represents all the nodes that precede $j$ in the topological ordering implied by $\D_k$. As each node is d-separated from its non-descendants given its parents, from the first condition of the general I-Markov property we obtain
	$$p_k(\bx) = \prod_{j=1}^q p_k(x_j\g \bx_{\pa_{\D_k}(j)}).$$ Moreover, each node $j \notin T^{(k)}$ is d-separated from $\zeta_{k}$ given its parents in $\D_k^\I$. Hence, from the second condition of the general I-Markov property we have $p_k(x_j\g\bx_{\pa_{\D_k}(j)}) = p(x_j\g \bx_{\pa_{\D}(j)})$, so that
	$$
	p_k(\bx) = \prod_{j\notin T^{(k)}}p(x_j\g\bx_{\pa_{\D}(j)})\prod_{j \in T^{(k)}} p_k(x_j\g \bx_{\pa_{\D_k}(j)}).
	$$
	Hence the result follows from the Lemma above. 
	
	\textit{Only if} - We want to prove that if $p_k(\cdot)$ factorizes according to
	$$p_k(\bx) = \prod_{j \notin T^{(k)}} p(x_{j}\g\bx_{\pa_{\D}(j)}) \prod_{j \in T^{(k)}} \tilde{p}(x_{j}\g\bx_{\pa_{\D_k}(j)})$$ for all $k \in [K]$, then the general I-Markov property holds, namely the collection of $\I$-DAGs $\{\D_k^{\I}\}_{k=1}^K$ can be used to recover all the conditional independencies and invariances through d-separation criteria. 
	
	As for the conditional independencies, note that by Lemma \ref{lemma:factorization} we have that $p_k(\cdot)$ factorizes according to $\D_k$ for all $k \in [K]$.
	Hence, for each $k \in [K]$ the Markov property defined on d-separation criteria must hold with respect to $\D_k$. Therefore, the first condition of the I-Markov property must hold.
	
	For the second condition, instead, we want to show that the invariant components of the distribution are exactly those whose nodes $j$'s are d-separated from $\zeta_I$ given a set $C$ in $\D_k^\I$, for all $k \in [K]$.
	Consider any two disjoint sets $A,C \subset [q]$ and $k \in [K]$ and suppose that $C$ d-separates $A$ from $\zeta_k$ in $\D_k^\I$. Now, let $V_{An}$ be the ancestral set of $A$ and $C$ in $\D_k$. Denote with $B' \subset V_{An}$ those nodes that are also d-connected to $\zeta_k$ in $\D_k^\I$ given $C$ and with $A' = V_{An} \backslash \{B' \cup C\}$ the sets of ancestors of $A$ and $C$ that are not d-connected to $\zeta_k$ and that are not in the conditioning set $C$. Note that $V_{An} = A' \cup B' \cup C$. From the factorization, we have that
	
	\begin{align*}
		p_k(\bx) & = p_k(\bx_{A'}, \bx_{B'}, \bx_C, \bx_{V \backslash V_{An}}) \\
		& = \prod_{j \in A'} p_k(x_j\g\bx_{\pa_{\D_k}(j)}) \prod_{j \in B'} p_k(x_j\g\bx_{\pa_{\D_k}(j)})  \\ 
		& \quad  \quad \prod_{j \in C} p_k(x_j\g\bx_{\pa_{\D_k}(j)}) \prod_{j \in V \backslash V_{An}} p_k(x_j\g\bx_{\pa_{\D_k}(j)}) \\ 
		& = \prod_{j \in A'} p_k(x_j\g\bx_{\pa_{\D_k}(j)}) \prod_{j \in B'} p_k(x_j\g\bx_{\pa_{\D_k}(j)}) \prod_{j \in C, \pa_{\D_k}(j) \cap A' = \O} p_k(x_j\g\bx_{\pa_{\D_k}(j)}) \\ 
		& \quad  \quad \prod_{j \in C, \pa_{\D_k}(j) \cap A' \neq \O} p_k(x_j\g\bx_{\pa_{\D_k}(j)}) \prod_{j \in V \backslash V_{An}} p_k(x_j\g\bx_{\pa_{\D_k}(j)}) \\ 
		& = \prod_{j \in A'} p(x_j\g\bx_{\pa_{\D}(j)}) \prod_{j \in B'} p_k(x_j\g\bx_{\pa_{\D_k}(j)}) \prod_{j \in C, \pa_{\D_k}(j) \cap A' = \O} p_k(x_j\g\bx_{\pa_{\D_k}(j)}) \\ 
		& \quad  \quad \prod_{j \in C, \pa_{\D_k}(j) \cap A' \neq \O} p(x_j\g\bx_{\pa_{\D}(j)}) \prod_{j \in V \backslash V_{An}} p_k(x_j\g\bx_{\pa_{\D_k}(j)}),
	\end{align*}
	where the last equality follows from the fact that 
	\begin{itemize}
		\item[$\cdot$] if $j \in A'$, then $j$ is d-separated from $\zeta_k$ in $\D_k^\I$ given $C$ and thus $j$ can not be a child of $\zeta_k$;
		\item[$\cdot$] if $j \in C$ and there exists at least one $h \in \pa_{\D_k}(j)$ such that $h \in A'$, then $j$ can not be a child of $\zeta_k$: if it were, then conditioning on $j$ its parents would be d-connected to $\zeta_k$ given $C$;
	\end{itemize}
	and recalling that $j \in \text{ch}_{\zeta_k}(\D_k^\I)$ if and only if $j \in T^{(k)}$. Similarly, the (union of) parents of nodes in $A'$ and $\{j \in C\g\pa_{\D_k}(j) \cap A' \neq \O\}$ are subsets of $A' \cup C$, while the parents of $B'$ and $\{j \in C\g\pa_{\D_k}(j) \cap A' = \O\}$ are subsets of $B' \cup C$. We can thus write 
	$$
	p_k(\bx) = g(\bx_{A'}, \bx_C)g_k(\bx_{B'}, \bx_C)g_k(\bx_{V \backslash V_{An}})
	$$
	just to underline the observational and interventional blocks in the factorization above and their arguments. We can thus marginalize out $A' \backslash A$, $B'$ and $V \backslash V_{An}$, thus obtaining
	\begin{align*}
		p_k(\bx_A, \bx_C) & = \int\limits_{ X_{(A' \backslash A) \cup B' \cup (V \backslash V_{An})}} g(\bx_{A'}, \bx_C)g_k( \bx_{B'}, \bx_C)g_k(\bx_{V \backslash V_{An}}) \\
		& = \int\limits_{ X_{(A' \backslash A) \cup B'}} g(\bx_{A'}, \bx_C)g_k(\bx_{B'},\bx_C) \\
		& = \int\limits_{ X_{(A' \backslash A)}} g( \bx_{A'}, \bx_C)\int\limits_{ X_{B'}}g_k(\bx_{B'},\bx_C) \\
		& = \tilde g (\bx_A, \bx_C)\tilde g_k(\bx_C).
	\end{align*}
	Using the latter expression we can write
	\begin{align*}
		p_k(\bx_A \g \bx_C) = \frac{p_k(\bx_A, \bx_C)}{p_k(\bx_C)} & = \frac{\tilde g ( \bx_A,\bx_C)\tilde g_k(\bx_C)}{\int_{X_A} \tilde g (\bx_A,\bx_C)\tilde g_k(\bx_C)} \\
		& = \frac{\tilde g (\bx_A,\bx_C)\tilde g_k(\bx_C)}{\tilde g_k(\bx_C)\int_{X_A} \tilde g (\bx_A,\bx_C)} \\
		& = \frac{\tilde g (\bx_A,\bx_C)}{\int_{X_A} \tilde g (\bx_A, \bx_C)},
	\end{align*}
	which does not depend on $T^{(k)}$ and is thus invariant as required by the Markov property.
\end{proof}

\begin{reptheorem}{theorem:imarkeq:skeleta}
	Let $\D_1, \D_2$ be two DAGs and $\I$ a collection of targets and induced parent
	sets inducing a valid general intervention for both $\D_1$ and $ \D_2$. $\D_1$ and $\D_2$ belong to the same I-Markov equivalence class if and only if $\D^\I_{1,k}$ and $\D^\I_{2,k}$ have the same skeleta and v-structures for all $k \in [K]$.
\end{reptheorem}

	

\begin{proof}
	\textit{If}: Because $\D_{1,k}^\I$ and $\D_{2,k}^\I$ have the same sleketon and set of v-structures for each $k \in [K]$, the two collections of $\I$-DAGs $\{\D^{\I}_{1,k}\}_{k=1}^K, \{\D^{\I}_{2,k}\}_{k=1}^K$ satisfy the same d-separation statements, thus implying the same sets of conditional independencies and invariances through the I-Markov property, so that $\mathcal{M}_{\I}(\D_1) = \mathcal{M}_{\I}(\D_2)$. \\
	
	\textit{Only if}: Suppose there exists a $k^* \in [K]$ such that $\D_{1,k^*}^{\I}$ and $\D_{2,k^*}^{\I}$ do not have the same skeleton and set of v-structures. Denote with $\D_{1,k^*}, \D_{2,k^*}$ the post intervention DAGs corresponding to the $k^*$th experimental setting. Note that $\D_{1,k^*}, \D_{2,k^*}$ have the same skeleta and sets of v-structures, otherwise $\D_{1,k^*}, \D_{2,k^*}$ would not be Markov equivalent and consequently $(\D_1, \D_2)$ would not be I-Markov equivalent given $\I$. Moreover, $\D_{1,k^*}^{\I}$ and $\D_{2,k^*}^{\I}$ have the same $\I$-edges, as these are determined by $T^{(k^*)}$. They thus differ for the sets of v-structures involving $\I$-edges. Suppose that $\zeta_{k^*} \to v \leftarrow w$ is a v-structure in  $\D_{1,k^*}^{\I}$, implying $w \notin T^{(k^*)}$ and $w \in P^{(k^*)}_v$, and that such v-structure is not present in $\D_{2,k^*}^{\I}$. As the modified DAGs $\D_{1,k^*}, \D_{2,k^*}$ have the same skeleton, then $\zeta_{k^*} \to v \to w \in \D_{2,k^*}^{\I}$. As the parent set of $v$ is fixed by the intervention, we would have that both $v \leftarrow w \in \D_{2,k}^\I$ and $v \rightarrow w \in \D_{2,k}^\I$, which implies a cycle and thus a contradiction with the validity assumption.
\end{proof}

We now focus on the transformational characterization of Theorem \ref{theorem:seqcovedge}.
\black

\begin{lemma}
	\label{lemma:simcovedge}
	Let $\D_1$ be a DAG containing the edge $u \to v$ and $\I$ a collection of targets and induced parent sets defining a general intervention. Let $\D_2$ be a graph identical to $\D_1$ except for the reversal of $u \to v$. $\D_1$ and $\D_2$ belong to the same I-Markov Equivalence class if and only if $u \to v$ is simultaneously covered;
	\begin{proof}
		\textit{If:}  Suppose $u \to v$ is simultaneously covered. Then, $u \to v$ is covered in $\D_1$ and, for any $k \neq 1$, $u \to v$ is either i) covered in $\D_{1,k}^\I$ or ii) $\{u,v\} \subseteq T^{(k)}$. In case i), we cannot have $u \in T^{(k)}$ and $v \notin T^{(k)}$ (or viceversa) by the definition of covered edge in the $\I$-DAG. The parent sets of the two nodes in the $\I$-DAGs are thus the same as in the observational DAG $\D$ and the proof follows from \citet[Lemma 1]{chickering:1995}. In case ii), both $u$ and $v$ are targets of intervention and reversing $u \to v$ in $\D_1$ does not cause any change in the parent sets of the nodes in the $\I$-DAGs. $u \to v$ thus has to be covered only in $\D$ and the proof follows again from \citet[Lemma 1]{chickering:1995}.
		
		\textit{Only if:} Suppose that $u \to v$ is not simultaneously covered. Then, at least one of the following statements is true: i) $u \to v$ is not covered in $\D_1$; ii) there exists $k^* \in [K]$ such that $u \to v$ is not covered in $\D_{1,k^*}^\I$ and $\{u,v\} \not\subseteq T^{(k^*)}$. In case i) the proof follows from \citet[Lemma 1]{chickering:1995}. In case ii), we have that, by the definition of a covered edge, $\pa_{\D_{1,k^*}^\I}(u) \cup u \neq \pa_{\D_{1,k^*}^\I}(v)$. In particular, either there exists at least one $z$ such that $z \in  \pa_{\D_{1,k^*}^\I}(u), z \notin  \pa_{\D_{1,k^*}^\I}(v)$, or there exists at least one node $w$ such that $w \in  \pa_{\D_{1,k^*}^\I}(v), w \notin  \pa_{\D_{1,k^*}^\I}(u)$. Consider the first case. Then, either (a) $z = \zeta_{k^*}$ or (b) $z \neq \zeta_{k^*}$. In case (a), note that $v \notin T^{(k^*)}$, by definition of $z$, so that $u \to v \in \D_{1,k^*}^\I$. As the intervention is defining the parent set of node $u$, we have that $\pa_{\D_{1,k^*}^\I}(u) = \pa_{\D_{2,k^*}^\I}(u)$. Moreover, the intervention is supposed to be valid, so that $v \notin \pa_{\D_{1,k^*}^\I}(u)$. We thus have that $u \to v \in \D_{1,k^*}^\I$, while both $u \to v, v \to u \notin \D_{2,k^*}^\I$. As $\D_{1,k^*}^\I, \D_{2,k^*}^\I$ differ for their skeleta, they can not be I-Markov equivalent. In case (b), instead, by the definition of a not simultaneously-covered edge, we have that $\zeta_{k^*}$ does not belong to the common parents of $\{u,v\}$. Hence, $\{u,v\} \not \subseteq T^{(k)}$ and $u \to v$ is covered in $\D_{1,k^*}^\I$ if and only if it is covered in $\D_1$ (and the same holds for $\D_2$). The proof thus follows from \citet[Lemma 1]{chickering:1995}. The proof for case $w \in  \pa_{\D_{1,k^*}^\I}(v), w \notin  \pa_{\D_{1,k^*}^\I}(u)$ follows by a similar reasoning.
	\end{proof}
\end{lemma}

Let $\Delta(\D_1, \D_2)$ denote the set of edges in $\D_1$ that have opposite orientation in $\D_2$ and $\Psi_v = \{u\g u \to v \in \Delta(\D_1, \D_2)\}$, the set of nodes that are parents of $v$ in $\D_1$ and children of $v$ in $\D_2$. Algorithm \ref{alg:findedge} was first presented in \citet{chickering:1995} to find a covered edge belonging to $\Delta(\D_1, \D_2)$ for two Markov Equivalent DAGs and it can be also adopted in our setting. 

\begin{algorithm}{
		\SetAlgoLined
		\vspace{0.1cm}
		\KwInput{DAGs $\D_1, \D_2$}
		\KwOutput{Edge from $\Delta(\D_1, \D_2)$}
		Perform a topological sort on the nodes in $\D_1$\;
		Let $v$ be the minimal node with respect to the sort for which $\Psi_v \neq \O$\;
		Let $u$ be the maximal node with respect to the sort for which $u \in \Psi_v$;
		
		\Return $u \to v$
	}
	\caption{Find-Edge \citep{chickering:1995}}
	\label{alg:findedge}
\end{algorithm}

\begin{lemma}
	\label{lemma:find:edge:1}
	Let $\D_1, \D_2$ be two I-Markov equivalent DAGs for $\I$, a collection of targets and induced parent sets defining a valid general intervention for both $\D_1, \D_2$.
	The edge $u \to v$ output from Algorithm \ref{alg:findedge} with input $\D_1, \D_2$ is simultaneously covered.
	\begin{proof}
		We know from Lemma 2 in \citet{chickering:1995} that $u \to v$ is covered in $\D_1$. 
		Suppose now that $u \to v$ is not simultaneously covered. Hence, there must exist at least one $k^* \neq 1$ such that $u \to v$ is not covered in $\D_{1,k^*}^\I$ and $\{u,v\} \not \subseteq T^{(k^*)}$. In particular, either i) $u \in T^{(k^*)}, v \notin T^{(k^*)}$ or ii) $v \in T^{(k^*)}, u \notin T^{(k^*)}$. 
		Suppose i). Note that $v \notin \pa_{\D_{1,k^*}^\I}(u)$ as the intervention is supposed to be valid. Hence, we have that $\zeta_{k^*} \rightarrow u \rightarrow v$ in $\D_{1,k^*}^\I$ and $\zeta_{k^*}\rightarrow u \not \leftarrow v$ in $\D_{2,k^*}^{\I}$. Because $\D_1, \D_2$ now differ for their skeleton in one of the $\I$-DAGs, they can not be I-Markov equivalent. 
		Suppose ii). In this case, we have that either (a) $u \notin \pa_{\D_{1,k^*}^\I}(v)$ or (b) $u \in \pa_{\D_{1,k^*}^\I}(v)$. In case (a), we have that $u \not \rightarrow v \leftarrow \zeta_{k^*}$ in $\D_{1,k^*}^\I$ and $u \leftarrow v \leftarrow \zeta_k$ in $\D_{2,k^*}^{\I}$, as the parents of $v$ remain invariant between $\D_2$ and $\D_{2,k^*}^{\I}$. The difference in skeleton implies that $\D_1, \D_2$ are not I-Markov equivalent, a contradiction. In case (b), for the same reason we would have $u \rightarrow v \leftarrow \zeta_{k^*}$ in $\D_{1,k^*}^\I$ and $u \leftrightarrow v \leftarrow \zeta_{k^*}$ in $\D_{2,k^*}^{\I}$ thus contradicting the fact that $\I$ is a valid collection of targets and induced parent sets. 
	\end{proof}
\end{lemma}

\begin{reptheorem}{theorem:seqcovedge}
	Let $\D_1, \D_2$ be two DAGs and $\I$ a collection of targets and induced parent sets defining a valid general intervention for both $\D_1$ and $\D_2$. $\D_1$ and $\D_2$ belong to the same I-Markov equivalence class if and only if there exists a sequence of edge reversals modifying $\D_1$ and such that:
	\begin{enumerate}
		\item Each edge reversed is simultaneously covered;
		\item After each reversal, $\{\D_{1,k}^{\I}\}_{k=1}^K$ are DAGs and $\D_1, \D_2$ belong to the same I-Markov equivalence class;
		\item After all reversals $\D_1 = \D_2$.
	\end{enumerate}
\end{reptheorem}

\begin{proof}
	\textit{If:} The proof follows immediately from the definition of the sequence. \\
	\textit{Only if:} We show that all the conditions are satisfied if we apply the procedure Find-Edge to $\D_1, \D_2$ to identify the next edge to reverse in $\D_1$. We know that $u \to v$, the output of Find-Edge, is a simultaneously covered edge (Lemma \ref{lemma:find:edge:1}). As it is simultaneously covered, the DAG obtained by reversing the edge still belongs to the same I-Markov equivalence class by Lemma \ref{lemma:simcovedge}. Moreover, $|\Delta(\D, \D^\prime)|$ decreases by one at each step. All the three conditions are thus satisfied. 
\end{proof}

\subsection*{A.2 Proofs of Section 2.3}
\label{appx:proofs:2:3}

We here report the proofs of the results presented in Section \ref{sec:identifiability:unknowninterventions}, concerning the identifiability of i) unknown general interventions and ii) unknown DAGs and general interventions. 

\begin{reptheorem}{theorem:imarkeq:skeleta2}
	Let $\D$ be a DAG and $\I_1, \I_2$ two collections of targets and induced parent sets. Then, $\I_1, \I_2$ belong to the same I-Markov equivalence class if and only if $\D^{\I_1}_{k}, \D^{\I_2}_{k}$ have the same skeleta and v-structures for all $k \in [K]$.
\end{reptheorem}

\begin{proof}
	\textit{If:}  As $\D_{k}^{\I_1}$ and $\D_k^{\I_2}$ have the same skeleton and same set of v-structures for all $k \in [K]$, they imply the same d-separation statements, thus implying the same sets of conditional independencies and invariances through the I-Markov property, so that $\M_{\I_1}(\D) = \M_{\I_2}(\D)$. \\
	\textit{Only if:} Suppose there exists $k^* \in [K]$ such that $\D_{k^*}^{\I_1}$ and $\D_{k^*}^{\I_2}$ do not have the same skeleton and set of v-structures. Denote with $\D_{1,k^*}, \D_{2,k^*}$ the post intervention DAGs corresponding to the $k^*$th experimental setting. Note that $\D_{1,k^*}, \D_{2,k^*}$ have the same skeleta and sets of v-structures, otherwise $\D_{1,k^*}, \D_{2,k^*}$ would not be Markov equivalent and consequently $\I_1, \I_2$ would not be I-Markov equivalent. $\D_{k^*}^{\I_1}$ and $\D_{k^*}^{\I_2}$ thus differ i) for their sets of $\I$-edges or ii) for v-structures involving the $\I$-edges. In case i), suppose without loss of generality that $\D_{k^*}^{\I_1}$ has an additional $\I$-edge $\zeta_{k^*} \to v$ which is not in $\D_{k^*}^{\I_2}$. Then $p_{k^*}(\bx_v \g \bx_{\pa_{\D_{1,k^*}}(v)}) \neq p_1(\bx_v \g \bx_{\pa_{\D_{1,k^*}}(v)})$, while $p_{k^*}(\bx_v \g \bx_{\pa_{\D_{2,k^*}}(v)}) = p_1(\bx_v \g \bx_{\pa_{\D_{2,k^*}}(v)})$ and $\I_1$, $\I_2$ can not be I-Markov equivalent. In case ii), suppose that $\zeta_{k^*} \to v \leftarrow w$ is a v-structure in  $\D_{k^*}^{\I_1}$, which implies $w \notin T^{(k^*)}_1$, and that such v-structure is not present in $\D_{k^*}^{\I_2}$. As the modified DAGs $\D_{1,k^*}, \D_{2,k^*}$ have the same skeleton, then $\zeta_{k^*} \to v \to w \in \D_{k^*}^{\I_2}$. However, because the parent set of $w$ is changing between the two DAGs and $w \notin T_1^{(k^*)}$, it means that $w \in T_2^{(k^*)}$, so that $\zeta_{k^*} \to w \in \D_{k^*}^{\I_2}$, inducing a difference in skeleton.
\end{proof}

We now focus on the transformational characterization of Theorem \ref{theorem:seqcovedge2}.  \black

\begin{lemma}
	\label{lemma:covered:targets}
	Let $\D$ be a DAG and $\I_1, \I_2$ two collections of targets and induced parent sets such that, for some $k \in [K]$, $\D_k^{\I_1}, \D_k^{\I_2}$ differ only for the reversal of $u \to v \in \D_{k}^{\I_1}$ becoming $v \to u \in \D_k^{\I_2}$. $\I_1$ and $\I_2$ belong to the same I-Markov equivalence class if and only if $u \to v$ is covered in $\D_k^{\I_1}$.
\end{lemma}
\begin{proof}
	\textit{If}: The proof is identical to \citet[Lemma 1]{chickering:1995}.
	
	\textit{Only if}: Notice that, by I-Markov equivalence, $\D_k^{\I_1}, \D_k^{\I_2}$ have the same skeleta and in particular the same $\I$-edges, so that $T_1^{(k)} = T_2^{(k)}$. Suppose now that $u \to v$ is not covered in $\D_k^{\I_1}$. Then $\pa_{\D_k^{\I_1}}(u) \cup u \neq \pa_{\D_k^{\I_1}}(v)$. In particular, either i) there exists some $z \in \pa_{\D_k^{\I_1}}(u), z \notin \pa_{\D_k^{\I_1}}(v)$ or ii) there exists some $w \in \pa_{\D_k^{\I_1}}(v), w \notin \pa_{\D_k^{\I_1}}(u)$. In case i), suppose that $z = \zeta_k$. In this case, $u \in T_1^{(k)}$ and $v \notin T_1^{(k)}$, so that $\pa_{\D_k^{\I_1}}(v) = \pa_\D(v)$. Because of the edge reversal, $\pa_{\D_k^{\I_1}}(v) \neq \pa_{\D_k^{\I_2}}(v)$, implying that $\pa_{\D_k^{\I_2}}(v) \neq \pa_\D(v)$ and $v \in T_2^{(k)}$, which is a contradiction as $T_1^{(k)} = T_2^{(k)}$. Hence, $z \neq \zeta_k$ and the proof follows from \citet[Lemma 1]{chickering:1995}. The proof for case ii) follows by a similar reasoning.
\end{proof} \black
\begin{lemma}
	\label{lemma:seq:edge:targets}
	Let $\D$ be a DAG and $\I_1, \I_2$ two collections of targets and induced parent sets belonging to the same I-Markov equivalence class. The edge $u \to v$ output from Algorithm \ref{alg:findedge} with input $\D_{k}^{\I_1}, \D_{k}^{\I_2}$ is covered.
\end{lemma}
\begin{proof}
	The proof is identical to the one of Lemma 2 in \citet{chickering:1995}. 
\end{proof}

\begin{reptheorem}{theorem:seqcovedge2}
	Let $\D$ be a DAG and $\I_1$, $\I_2$ two collection of targets and induced parent sets. Then, $\I_1$, $\I_2$ belong to the same I-Markov equivalence class if and only if for each $\I$-DAG $\D_k^{\I_1}$, $k \neq 1$, there exists a sequence of edge reversals modifying $\D_k^{\I_1}$ and such that:
	\begin{enumerate}
		\item Each edge reversed is covered;
		\item After each reversal, $\D_{k}^{\I_1}$ is a DAG and $\I_1$, $\I_2$ belong to the same I-Markov equivalence class;
		\item After all reversals $\D_{k}^{\I_1} = \D_{k}^{\I_2}$.
	\end{enumerate}
\end{reptheorem}

\begin{proof}
	\textbf{If: } It follows immediately from the definition of the sequence. \\
	\textbf{Only if: } We show that all the conditions are satisfied if we apply the procedure Find-Edge with input $\D_k^{\I_1}$ and $\D_k^{\I_2}$, for all $k \neq 1$. We know that $u \to v$, output of Find-Edge is covered (Lemma \ref{lemma:seq:edge:targets}) and that the $\I$-DAG obtained by reversing $u \to v$ corresponds to a collection of targets and induced parent sets which is I-Markov equivalent to the initial one (Lemma \ref{lemma:covered:targets}). At each step, $\Delta(\D_k^{\I_1}, \D_k^{\I_2})$ decreases by one. All the three conditions are thus satisfied. 
\end{proof}

We now consider the set of results concerning the joint identifiability of a pair $(\D, \I)$. \black

\begin{reptheorem}{theorem:imarkeq:skeleta3}
	Let $\D_1, \D_2$ be two DAGs and $\I_1, \I_2$ two collections of targets and induced parent sets defining a valid general intervention for $\D_1, \D_2$ respectively. $(\D_1, \I_1), (\D_2, \I_2)$ belong to the same I-Markov equivalence class if and only if $\D^{\I_1}_{1,k}, \D^{\I_2}_{2,k}$ have the same skeleta and v-structures for all $k \in [K]$.
\end{reptheorem}

\begin{proof}
	\textit{If: } As $\D_{k}^{\I_1}$ and $\D_k^{\I_2}$ have the same skeleta and set of v-structures for all $k \in [K]$, they imply the same d-separation statements, thus implying the same sets of conditional independencies and invariances through the I-Markov property, so that $\M_{\I_1}(\D) = \M_{\I_2}(\D)$. \\
	\textit{Only if:} Suppose there exists $k^* \in [K]$ such that $\D_{1,k^*}^{\I_1}$ and $\D_{2,k^*}^{\I_2}$ do not have the same skeleton and set of v-structures. Denote with $\D_{1,k^*}, \D_{2,k^*}$ the post intervention DAGs corresponding to the $k^*$th experimental setting. Note that $\D_{1,k^*}, \D_{2,k^*}$ have the same skeleta and sets of v-structures, otherwise $\D_{1,k^*}, \D_{2,k^*}$ would not be Markov equivalent and consequently $(\D_1, \I_1)$, $(\D_2, \I_2)$ would not be I-Markov equivalent. $\D_{1,k^*}^{\I_1}$ and $\D_{2,k^*}^{\I_2}$ thus differ i) for their sets of $\I$-edges or ii) for v-structures involving the $\I$-edges. In case i), suppose without loss of generality that $\D_{1,k^*}^{\I_1}$ has an additional $\I$-edge $\zeta_{k^*} \to v$ which is not in $\D_{2,k^*}^{\I_2}$. Then $p_{k^*}(\bx_v \g \bx_{\pa_{\D_{1,k^*}}(v)}) \neq p_1(\bx_v \g \bx_{\pa_{\D_{1,k^*}}(v)})$, while $p_{k^*}(\bx_v \g \bx_{\pa_{\D_{2,k^*}}(v)}) = p_1(\bx_v \g \bx_{\pa_{\D_{2,k^*}}(v)})$ and $(\D_1, \I_1)$, $(\D_2, \I_2)$ can not be I-Markov equivalent. In case ii), suppose that $\zeta_{k^*} \to v \leftarrow w$ is a v-structure in  $\D_{1,k^*}^{\I_1}$ which is not present in $\D_{2,k^*}^{\I_2}$. As the modified DAGs $\D_{1,k^*}, \D_{2,k^*}$ have the same skeleton, then $\zeta_{k^*} \to v \to w \in \D_{2,k^*}^{\I_2}$. We thus have that $w$ is d-separated from $\zeta_{k^*}$ in $\D_{1,k}^{\I_1}$, but not in $\D_{2,k}^{\I_2}$. By the I-Markov property, it follows that $p_{k^*}(\bx_w
	) = p_1(\bx_w
	)$, while $p_{k^*}(\bx_w
	) \neq p_1(\bx_w
	)$ and $(\D_1, \I_1)$, $(\D_2, \I_2)$ can not be I-Markov equivalent.
\end{proof} \black

\begin{lemma}
	Let $\D_1, \D_2$ be two DAGs and $\I_1, \I_2$ two collections of targets and induced parent sets defining a valid general intervention for both $\D_1, \D_2$. Suppose in addition that $(\D_1, \I_1)$, $(\D_2, \I_2)$  differ only for the reversal of $u \to v \in \D_1$ becoming $v \to u \in \D_2$. $(\D_1, \I_1), (\D_2, \I_2)$ belong to the same I-Markov equivalence class if and only if $u \to v$ is simultaneously covered in $\D_1$.
	\label{lemma:coveredsimtargets1}
\end{lemma}

\begin{proof}
	By construction, we have that $\I_1 = \I_2$. Consequently, the statement and its proof coincide with those of Lemma \ref{lemma:simcovedge}.
\end{proof}

\begin{lemma}
	Let $\D_1, \D_2$ be two DAGs and $\I_1, \I_2$ two collections of targets and induced parent sets defining a valid general intervention for both $\D_1, \D_2$. Suppose in addition that $(\D_1, \I_1)$, $(\D_2, \I_2)$  differ only for the reversal of $u \to v \in \D_{1,k^*}^{\I_1}$ becoming $v \to u \in \D_{2,k^*}^{\I_2}$, for some $k^* \neq 1$. $(\D_1, \I_1), (\D_2, \I_2)$ belong to the same I-Markov equivalence class if and only if $u \to v$ is covered in $\D_{1,k}^{\I_1}$.
	\label{lemma:coveredsimtargets2}
\end{lemma}
\begin{proof}
	By construction, $\D_1 = \D_2$. Consequently, the statement and its proof coincide with those of Lemma \ref{lemma:covered:targets}.
\end{proof}

\begin{reptheorem}{theorem:seqcovedge3}
	Let $\D_1, \D_2$ be two DAGs and $\I_1$, $\I_2$ two collections of targets and induced parent sets defining a valid general intervention for both $\D_1, \D_2$. $(\D_1, \I_1)$, $(\D_2, \I_2)$ belong to the same I-Markov equivalence class if and only if there exists a sequence of edge reversals modifying the collection of $\I$-DAGs $\{\D_{1,k}^{\I_1}\}_{k=1}^K$ and such that:
	\begin{enumerate}
		\item Each edge reversed in $\D_1$ is simultaneously covered;
		\item Each edge reversed in $\D_{1,k}^{\I_1}$, for $k \neq 1$, is covered;
		\item After each reversal, $\{\D_{1,k}^{\I_1}\}_{k=1}^K$ are DAGs and $(\D_1,\I_1)$, $(\D_2, \I_2)$ belong to the same I-Markov equivalence class;
		\item After all reversals $\D_{1,k}^{\I_1} = \D_{2,k}^{\I_2}$ for each $k \in [K]$.
	\end{enumerate}
\end{reptheorem}
\begin{proof}
	One can construct a sequence of edge reversals satisfying all the conditions by first using Algorithm \ref{alg:findedge} with inputs $\D_{1,k}^{\I_1}, \D_{1,k}^{\I_2}$ for $k \in [K], k \neq 1$, and then using the same Algorithm with inputs $\D_1, \D_2$. For each of these two steps, the proofs follow the ones of the corresponding Theorems \ref{theorem:seqcovedge} and \ref{theorem:seqcovedge2}, using Lemmas \ref{lemma:coveredsimtargets1} and \ref{lemma:coveredsimtargets2}.
\end{proof}

\black


\section*{Appendix B. Proofs of Section 3}

This section contains the proofs of the main results presented in Section 3 of
the paper. The numbering of such propositions and theorems in this section is the same as
in the main text. 

\begin{repproposition}{prop:marglik}
	Given any complete DAG $C$ and a data matrix $\bX$ collecting observations from $K$ different experimental settings, for any valid pair $(\D, \I)$ Assumptions \textbf{A1}-\textbf{A3} imply
	\begin{align}
		\label{eq:marg:like:general}
		\begin{split}
			p\left(\bX \g \D, \I \right) & =
			\prod_{j=1}^q \left\{ \frac{p\left(\boldsymbol{X}_{\cdot \fa_{\D}(j)}^{\mathcal{A}(j)} \g C \right)}{p\left(\boldsymbol{X}_{\cdot \pa_{\D}(j)}^{\mathcal{A}(j)} \g C \right)} \prod_{k: j \in T^{(k)}} \frac{p\left(\boldsymbol{X}_{\cdot \fa_{\D_k}(j)}^{(k)} \g C \right)}{p\left(\boldsymbol{X}_{\cdot \pa_{\D_k}(j)}^{(k)} \g C \right)} \right\},
		\end{split}
	\end{align}
	where $p\big(\boldsymbol{X}_{\cdot B}^{\mathcal{A}(j)} \g C \big)$ is the marginal data distribution computed under any complete DAG $C$.
\end{repproposition}

\begin{proof}
	Using Equations \eqref{eq:marglikint} and \eqref{eq:likelihoodfinal}, together with Assumption \textbf{A3}, we can write
	\begin{align*}
		p\big(\bX\g\D, \I\big) & = \int{p\big(\bX \g \Theta^{(\mathcal{K})}, \D, \I\big)\,p\big(\Theta^{(\mathcal{K})}\g\D, \I\big)\ d\Theta^{(\mathcal{K})}}\\
		& = \int \prod_{j=1}^q \Bigg\{ p\left(\boldsymbol{X}_{\cdot j}^{\mathcal{A}(j)} \g \boldsymbol{X}_{\cdot\pa_{\D}(j)}^{\mathcal{A}(j)}, \Theta_j^{(1)}, \D \right) 
		\prod_{k: j \in T^{(k)}} p\left(\boldsymbol{X}_{\cdot j}^{(k)} \g \boldsymbol{X}_{\cdot\pa_{\D_k}(j)}^{(k)}, \Theta_j^{(k)}, \D_k \right)  \\ & \qquad \qquad \quad \quad p\left(\Theta_j^{(1)}\g\D\right)\prod_{k:j \in T^{(k)}} p\left(\Theta_j^{(k)}\g\D_k\right)\Bigg\} \ d\Theta^{(\mathcal{K})} \\
		& = \prod_{j=1}^q \Bigg\{ \int p\left(\boldsymbol{X}_{\cdot j}^{\mathcal{A}(j)} \g \boldsymbol{X}_{\cdot\pa_{\D}(j)}^{\mathcal{A}(j)}, \Theta_j^{(1)}, \D \right) p\left(\Theta_j^{(1)}\g\D\right) d \Theta_j^{(1)} \\ & \qquad \qquad \prod_{k: j \in T^{(k)}} \int p\left(\boldsymbol{X}_{\cdot j}^{(k)} \g \boldsymbol{X}_{\cdot\pa_{\D_k}(j)}^{(k)}, \Theta_j^{(k)}, \D_k \right)p\left(\Theta_j^{(k)}\g\D_k\right)\ d\Theta^{(\mathcal{K})} \Bigg\}.
	\end{align*}
	By Assumption \textbf{A2} (likelihood and prior modularity), it follows that
	\begin{align*}
		p\big(\bX\g\D,\I\big)
		&  = \prod_{j=1}^q \Bigg\{ \int  p\left(\boldsymbol{X}_{\cdot j}^{\mathcal{A}(j)} \g \boldsymbol{X}_{\cdot\pa_{C_j}(j)}^{\mathcal{A}(j)}, \Theta_j^{(1)}, C_j \right) p\left(\Theta_j^{(1)}\g C_j\right) \ d\Theta^{(1)} \\ & \qquad \qquad \prod_{k: j \in T^{(k)}} \int p\left(\boldsymbol{X}_{\cdot j}^{(k)} \g \boldsymbol{X}_{\cdot\pa_{C_{j,k}}(j)}^{(k)}, \Theta_j^{(k)}, C_{j,k} \right)p\left(\Theta_j^{(k)}\g C_{j,k}\right) \ d\Theta^{(k)} \Bigg\} \\
		& = \prod_{j=1}^q \Bigg\{ p\left(\boldsymbol{X}_{\cdot j}^{\mathcal{A}(j)} \g \boldsymbol{X}_{\cdot\pa_{C_j}(j)}^{\mathcal{A}(j)}, C_j \right) \prod_{k:j\in T^{(k)}} p\left(\boldsymbol{X}_{\cdot j}^{(k)} \g \boldsymbol{X}_{\cdot\pa_{C_{j,k}}(j)}^{(k)}, C_{j,k} \right) \Bigg\}.
	\end{align*}
	Now by Assumption \textbf{A1} (complete model equivalence) and recalling that $\pa_{C_j}(j) = \pa_\D(j)$ and $\pa_{C_{j,k}}(j) = \pa_{\D_k}(j)$, we obtain
	\begin{align*}
		p\big(\bX\g\D,\I\big) &  = \prod_{j=1}^q \Bigg\{ p\left(\boldsymbol{X}_{\cdot j}^{\mathcal{A}(j)} \g \boldsymbol{X}_{\cdot\pa_{\D}(j)}^{\mathcal{A}(j)}, C \right) \prod_{k:j\in T^{(k)}} p\left(\boldsymbol{X}_{\cdot j}^{(k)} \g \boldsymbol{X}_{\cdot\pa_{\D_k}(j)}^{(k)}, C \right) \Bigg\} \\
		& = \prod_{j=1}^q \left\{ \frac{p\left(\boldsymbol{X}_{\cdot \fa_{\D}(j)}^{\mathcal{A}(j)} \g C \right)}{p\left(\boldsymbol{X}_{\cdot \pa_{\D}(j)}^{\mathcal{A}(j)} \g C \right)} \prod_{k: j \in T^{(k)}} \frac{p\left(\boldsymbol{X}_{\cdot \fa_{\D_k}(j)}^{(k)} \g C \right)}{p\left(\boldsymbol{X}_{\cdot \pa_{\D_k}(j)}^{(k)} \g C \right)} \right\},
	\end{align*}
	which completes the proof.
\end{proof}

\begin{reptheorem}{thm:scoreequivalence}[Score equivalence]
	Let $\D_1, \D_2$ be two DAGs and $\I_1, \I_2$ two collections of targets and induced parent sets defining a valid general intervention for $\D_1, \D_2$ respectively. If $(\D_1,\I_1)$ and $(\D_2,\I_2)$ are I-Markov equivalent, then Assumptions A1-A3 imply 
	\begin{equation}
		p(\bX\g\D_1, \I_1) = p(\bX\g\D_2, \I_2).
	\end{equation}
\end{reptheorem}

\begin{proof}
	By Theorem \ref{theorem:seqcovedge3}, there exists a sequence of edge reversals applied to either $\D_1$ or $\D_{1,k}^I, k \neq 1$ and such that, at the end of the sequence $(\D_1, \I_1) = (\D_2, \I_2)$. Let for simplicity $(\D, \I)$ be the pair of DAG and collection of targets and induced parent sets obtained at a given step of the sequence. We can consider the Bayes factor between $(\D, \I)$ and $(\widetilde \D, \widetilde \I)$, the corresponding pair obtained at the subsequent step. These two pairs differ for either i) a simultaneously covered edge reversal or ii) a covered edge reversal in one of the $\I$-DAGs $\D_k^\I, k \neq 1$. In case i), suppose that $\D, \widetilde \D$ differ for the simultaneously covered edge $u \to v \in \D$, which is reversed in $\widetilde \D$, while $\I = \widetilde \I$. Then 
	\begin{align*}
		\frac{p\big(\bX\g\D, \I\big)}{p\big(\bX\g\widetilde\D, \widetilde\I\big)}
		& = \left(\prod_{j=1}^q \left\{ \frac{p\left(\boldsymbol{X}_{\cdot \fa_{\D}(j)}^{\mathcal{A}(j)} \g C \right)}{p\left(\boldsymbol{X}_{\cdot \pa_{\D}(j)}^{\mathcal{A}(j)} \g C \right)} \prod_{k: j \in T^{(k)}} \frac{p\Big(\boldsymbol{X}_{\cdot \fa_{\D_{1,k}}(j)}^{(k)} \g C \Big)}{p\Big(\boldsymbol{X}_{\cdot \pa_{\D_{1,k}}(j)}^{(k)} \g C \Big)} \right\}\right)\cdot \\
		& \, \cdot \,\, \left(\prod_{j=1}^q \left\{ \frac{p\left(\boldsymbol{X}_{\cdot \fa_{\widetilde\D}(j)}^{\mathcal{A}(j)} \g C \right)}{p\left(\boldsymbol{X}_{\cdot \pa_{\widetilde\D}(j)}^{\mathcal{A}(j)} \g C \right)} \prod_{k: j \in \widetilde T^{(k)}} \frac{p\Big(\boldsymbol{X}_{\cdot \fa_{\widetilde\D_k}(j)}^{(k)} \g C \Big)}{p\Big(\boldsymbol{X}_{\cdot \pa_{\widetilde\D_k}(j)}^{(k)} \g C \Big)} \right\}\right)^{-1} \\
		& = \left(\prod_{j=1}^q \frac{p\left(\boldsymbol{X}_{\cdot \fa_{\D}(j)}^{\mathcal{A}(j)} \g C \right)}{p\left(\boldsymbol{X}_{\cdot \pa_{\D}(j)}^{\mathcal{A}(j)} \g C \right)}\right)\cdot \left(\prod_{j=1}^q \frac{p\left(\boldsymbol{X}_{\cdot \fa_{\widetilde\D}(j)}^{\mathcal{A}(j)} \g C \right)}{p\left(\boldsymbol{X}_{\cdot \pa_{\widetilde\D}(j)}^{\mathcal{A}(j)} \g C \right)}\right)^{-1} \\
		& = \left(\frac{p\left(\boldsymbol{X}_{\cdot \fa_{\D}(u)}^{\mathcal{A}(u)} \g C \right)}{p\left(\boldsymbol{X}_{\cdot \pa_{\D}(u)}^{\mathcal{A}(u)} \g C \right)}\frac{p\left(\boldsymbol{X}_{\cdot \fa_{\D}(v)}^{\mathcal{A}(v)} \g C \right)}{p\left(\boldsymbol{X}_{\cdot \pa_{\D}(v)}^{\mathcal{A}(v)} \g C \right)}\right) \cdot
		\left(\frac{p\left(\boldsymbol{X}_{\cdot \fa_{\widetilde\D}(u)}^{\mathcal{A}(v)} \g C \right)}{p\left(\boldsymbol{X}_{\cdot \pa_{\widetilde\D}(u)}^{\mathcal{A}(v)} \g C \right)}\frac{p\left(\boldsymbol{X}_{\cdot \fa_{\widetilde\D}(v)}^{\mathcal{A}(v)} \g C \right)}{p\left(\boldsymbol{X}_{\cdot \pa_{\widetilde\D}(v)}^{\mathcal{A}(v)} \g C \right)}\right)^{-1}.
	\end{align*}
	Because $\D$ and $\widetilde\D$ differ for the reversal of the simultaneously covered edge $u\rightarrow v$, then the following equalities holds:
	\begin{equation}
		\label{eq:set:equal:covered:edge}
		\pa_{\D}(u)=\pa_{\widetilde\D}(v), \quad \fa_{\D}(v)=\fa_{\widetilde\D}(u), \quad 
		\fa_{\D}(u)=\pa_{\D}(v), \quad \fa_{\widetilde\D}(v)=\pa_{\widetilde\D}(u).
	\end{equation}
	Therefore, the ratio simplifies to 1 if $A(u) = A(v)$. To prove this, notice that for any $j \in [q]$
	\begin{align*}
		\A(j)  \coloneqq & \, \big\{k \in [K]: j \notin T^{(k)}\big\} \\
		= & \, \big\{k \in [K]: \zeta_k \notin \pa_{\D_k^\I}(j)\big\}.
	\end{align*}
	Suppose now $\mathcal{A}(u) \neq \mathcal{A}(v)$. As a consequence, there exists $k \in [K]$ such that $\zeta_k \in \pa_{\D_k^\I}(u)$, while $\zeta_k \notin \pa_{\D_k^\I}(v)$, or viceversa.
	In both cases, this however would imply that $u \to v$ is not simultaneously covered, which is a contradiction, and therefore $\mathcal{A}(u) = \mathcal{A}(v)$.
	In case ii), suppose that, for some $k \in [K]$, $\D_k^\I, \widetilde \D_k^\I$ differ for the covered edge $u \to v \in \D_k^\I$ , which is reversed in $\widetilde \D_{k}^\I$.  Then $\D = \widetilde \D$ and
	\begin{align*}
		\frac{p\big(\bX\g \D, \I\big)}{p\big(\bX\g\widetilde\D, \widetilde\I\big)}
		& = \left(\prod_{j=1}^q \left\{ \frac{p\left(\boldsymbol{X}_{\cdot \fa_{\D}(j)}^{\mathcal{A}(j)} \g C \right)}{p\left(\boldsymbol{X}_{\cdot \pa_{\D}(j)}^{\mathcal{A}(j)} \g C \right)} \prod_{k: j \in T^{(k)}} \frac{p\Big(\boldsymbol{X}_{\cdot \fa_{\D_k}(j)}^{(k)} \g C \Big)}{p\Big(\boldsymbol{X}_{\cdot \pa_{\D_k}(j)}^{(k)} \g C \Big)} \right\}\right)\cdot \\
		& \, \cdot \,\, \left(\prod_{j=1}^q \left\{ \frac{p\left(\boldsymbol{X}_{\cdot \fa_{\widetilde\D}(j)}^{\mathcal{A}(j)} \g C \right)}{p\left(\boldsymbol{X}_{\cdot \pa_{\widetilde\D}(j)}^{\mathcal{A}(j)} \g C \right)} \prod_{k: j \in \widetilde T^{(k)}} \frac{p\Big(\boldsymbol{X}_{\cdot \fa_{\widetilde\D_k}(j)}^{(k)} \g C \Big)}{p\Big(\boldsymbol{X}_{\cdot \pa_{\widetilde\D_k}(j)}^{(k)} \g C \Big)} \right\}\right)^{-1} \\
		& = \left(\frac{p\left(\boldsymbol{X}_{\cdot \fa_{\D_k}(u)}^{(k)} \g C \right)}{p\left(\boldsymbol{X}_{\cdot \pa_{\D_k}(u)}^{(k)} \g C \right)}\frac{p\left(\boldsymbol{X}_{\cdot \fa_{\D_k}(v)}^{(k)} \g C \right)}{p\left(\boldsymbol{X}_{\cdot \pa_{\D_k}(v)}^{(k)} \g C \right)}\right) \cdot
		\left(\frac{p\Big(\boldsymbol{X}_{\cdot \fa_{\widetilde\D_k}(u)}^{(k)} \g C \Big)}{p\Big(\boldsymbol{X}_{\cdot \pa_{\widetilde\D_k}(u)}^{(k)} \g C \Big)}\frac{p\Big(\boldsymbol{X}_{\cdot \fa_{\widetilde\D_k}(v)}^{(k)} \g C \Big)}{p\Big(\boldsymbol{X}_{\cdot \pa_{\widetilde\D_k}(v)}^{(k)} \g C \Big)}\right)^{-1}
	\end{align*}
	where the second equality follows from the fact that
	by the I-Markov equivalence of $\I$ and $\widetilde\I$,
	$\T=\widetilde \T$.
	Since $u \rightarrow v$ is covered in the two DAGs, the equalities in \eqref{eq:set:equal:covered:edge} still hold and the ratio simplifies to 1.
\end{proof}

\section*{Appendix C. Proofs of Section 4}

This section contains the proof of Proposition \ref{prop:markov:chain} which establishes the convergence of Algorithms \ref{alg:mhwithingibbs} to the posterior distribution $p(\D, \T, \mathcal{P} \g \bX)$.

\begin{repproposition}{prop:markov:chain}
	The finite Markov chain defined by Algorithm \ref{alg:mhwithingibbs}, \ref{alg:validopD}, and \ref{alg:validopDkI} is reversible, aperiodic, and irreducible. Accordingly, it has $p(\D, \T, \mathcal{P} \g \bX)$ as its unique stationary distribution.
\end{repproposition}

\begin{proof}
	The reversibility and aperiodicity of Algorithm \ref{alg:mhwithingibbs} follows immediately from the properties of the Metropolis-Hastings algorithm \citep{Craiu:Rosenthal:2014} To prove irreducibility, notice that if, at each step of the Markov chain, both i) $p(\tilde\D, \I \g \bX)$ and ii) the proposal ratio are strictly greater than zero, then evaluating the irreducibility of Algorithm 1 reduces to evaluating the irreducibility of the Markov chain defined by the proposal distribution, illustrated in Algorithm \ref{alg:mhonlyprop}.
	%
	\begin{algorithm}{
			\SetAlgoLined
			\vspace{0.1cm}
			\KwInput{Number of iterations $S$, initial values for DAG, targets and induced parent sets $\D^0,\T^0,\mathcal{P}^0$}
			\KwOutput{A sample from a Markov chain over $(\D, \T, \mathcal{P})$}
			Construct $\big\{{\D_k^{s}}^\I\big\}_{k=1}^K$\;
			Set $\I^0=(\T^0,\mathcal{P}^0)$\; 
			\For{s in 1:S}{
				Sample $\boldsymbol \pi$, a permutation vector of length $K$\;
				Set $\{\D^s, \I^s\} = \{\D^{s-1}, \I^{s-1}\}$\;
				\For{$k$ in 1:K}{
					\If{$\boldsymbol{\pi}_k = 1$}{
						Construct $\cO_{\D^{s}}$ using Algorithm \ref{alg:validopD}\;
						Sample $\widetilde \D$ uniformly at random from $\cO_{\D^{s}}$\;
						Set $\D^s = \widetilde\D$
				}
				\Else{
					Construct $\cO_{{\D_{\boldsymbol \pi_k}^{s \I}}}$ using Algorithm \ref{alg:validopDkI}\;
					Sample $\widetilde \D_{\boldsymbol{\pi}_k}^\I$
					uniformly at random from $\cO_{{\D_{\boldsymbol{\pi}_k}^{s \I}}}$\;
					Recover $\widetilde I^{(\boldsymbol{\pi}_k)} = (\widetilde T^{({\boldsymbol{\pi}_k})}, \widetilde P^{({\boldsymbol{\pi}_k})})$ from  $(\widetilde \D_{\boldsymbol{\pi}_k}^\I, \D^s)$\;
					Set $I_s^{(\boldsymbol{\pi}_k)} = \widetilde I^{(\boldsymbol{\pi}_k)}$
				}
			}
		}
		Recover $\{\T^s, \mathcal{P}^s\}_{s=1}^S$ from $\{\I^s\}_{s=1}^S$\;
		\Return $\{\D^s, \T^s, \mathcal{P}^s\}_{s=1}^S$;
	}
	\caption{Markov chain implied by the proposal distribution of Algorithm \ref{alg:mhwithingibbs}\black}
	\label{alg:mhonlyprop}
\end{algorithm}
Requirement i) is trivially satisfied in the case of priors on $(\D, \I)$ with full support, as both the proposal distributions defined by Algorithm \ref{alg:validopD} and \ref{alg:validopDkI} explicitly take into account the validity requirement while defining the set of possible operators.
Condition ii) is satisfied if each move in the Markov chain is invertible, that is $q(\tilde{\D} \g \D) > 0$ if and only if $q(\D \g \tilde{\D}) > 0$.
Because of the structure of our proposal distributions in Algorithms \ref{alg:validopD} (\ref{alg:validopDkI}) this is equivalent to establish for each type of operator the existence of an \emph{inverse} operator; specifically,
we need to prove that if an operator belongs to $\cO_\D$ ($\cO_{\D_k^\I}$), then its inverse operator belongs to $\cO_{\tilde{\D}}$ ($\cO_{\tilde{\D}_k^\I}$) too.
For $\cO_\D$, whose construction is based on operators ${Insert}(u,v)$, ${Delete}(u,v)$ and ${Reverse}(u,v)$ applied to $u,v\in [q], u \neq v$, the proof is immediate: ${Insert}(u,v)$ is the inverse operator of ${Delete}(u,v)$ and viceversa, while ${Reverse}(u,v)$ is the inverse operator of ${Reverse}(v,u)$. The same holds when 
the three operators are applied to $u,v \in [q]$ for the construction of $\cO_{{\D}_k^\I}$.
In addition, when operators $Insert$ and $Delete$ involve $\zeta_k$, we have $Insert(\zeta_k,v)$ as the inverse operator of $Delete(\zeta_k,v)$ and viceversa.

We can thus prove the irreducibility of the chain defined by Algorithm \ref{alg:mhwithingibbs} by proving the irreducibility of the Markov chain defined by Algorithm \ref{alg:mhonlyprop}.
At each step $s$ of the algorithm, the proposed value is accepted and the new sequence of $\I$-DAGs $\{\D_{s,k}^\I\}_{k=1}^K$ is obtained by
sequentially updating each $\I$-DAG in a random order defined by the random permutation $\pi_s$. 
Notice that each component-wise update is reversible as shown before.
Moreover, any permutation vector $\pi$
admits an inverse permutation vector.
Therefore, to prove the irreducibility of \ref{alg:mhonlyprop}, it is sufficient to note that starting from any DAG $\{\tilde\D_k^\I\}$, it is always possible to reach the sequence of empty augmented DAGs $\{\bar \D_k^\I\}_{k=1}^K$ by repeated edge deletions. By reversibility, this implies that it is always possible to reach any DAG starting from any other DAG.
As the irreducibility of \ref{alg:mhonlyprop} implies the irreducibility of \ref{alg:mhwithingibbs}, the result follows. 
\end{proof}

\bibliographystyle{biometrika} 
\bibliography{biblio}

\end{document}